\documentclass[envcountsame]{llncs}

\usepackage{amssymb, amsmath,graphicx,graphics,enumerate, tkz-graph}
\usepackage[hypertexnames=false]{hyperref}
\usepackage[T1]{fontenc}
\usepackage[utf8]{inputenc}
\usepackage{lmodern}
\usepackage{microtype}

\DeclareMathOperator{\net}{net}
\DeclareMathOperator{\bull}{bull}

\DeclareMathOperator{\house}{house}

\DeclareMathOperator{\tw}{tw}
\DeclareMathOperator{\cw}{cw}

\newcommand{\ssi}{\subseteq_i}

\newcommand{\NP}{{\sf NP}}

\newcounter{ctrclaim}[theorem]
\newcounter{ctrcase}[theorem]
\newcounter{ctrsubcase}[ctrcase]
\renewcommand{\thectrsubcase}{\thectrcase\alph{ctrsubcase}}
\makeatletter
\@addtoreset{ctrclaim}{lemma} 
\@addtoreset{ctrcase}{lemma}
\makeatother
\newcounter{ctrfake}
\makeatletter
\@addtoreset{ctrclaim}{ctrfake} 
\@addtoreset{ctrcase}{ctrfake}
\makeatother

\newenvironment{enumeratei}{\begin{enumerate}[(i)]}{\end{enumerate}}

\newcommand\faketheorem[1]{\refstepcounter{ctrfake}{\bf #1}}
\newcommand\displaycase[1]{{\bf #1}}

\newcommand{\clm}[1]{\medskip\phantomsection\refstepcounter{ctrclaim}\noindent\displaycase{Claim \thectrclaim. }{\em #1}\\}
\newcommand{\thmcase}[1]{\medskip\phantomsection\refstepcounter{ctrcase}\noindent\displaycase{Case \thectrcase: }{\em #1}\\}
\newcommand{\thmsubcase}[1]{\medskip\phantomsection\refstepcounter{ctrsubcase}\noindent\displaycase{Case \thectrsubcase: }{\em #1}\\}

\title{Bounding the Clique-Width of\\ ${H}$-free Chordal Graphs\thanks{
An extended abstract of this paper appeared in the proceedings of MFCS 2015~\cite{BDHP15}.
The research in this paper was supported by EPSRC (EP/K025090/1).
The third author is grateful for the generous support of the Graduate (International) Research Travel Award from Simon Fraser University and Dr. Pavol Hell's NSERC Discovery Grant.}}
\author{Andreas Brandst{\"a}dt\inst{1} \and Konrad K. Dabrowski\inst{2} \and\\ Shenwei Huang\inst{3} \and Dani\"el Paulusma\inst{2}}
\institute{
Institute of Computer Science, Universität Rostock,\\
Albert-Einstein-Straße 22, 18059 Rostock, Germany\\
\texttt{ab@informatik.uni-rostock.de}
\and
School of Engineering and Computing Sciences, Durham University,\\
Science Laboratories, South Road,
Durham DH1 3LE, United Kingdom
\texttt{\{konrad.dabrowski,daniel.paulusma\}@durham.ac.uk}
\and
School of Computing Science, Simon Fraser University,\\
8888 University Drive, Burnaby B.C., V5A 1S6, Canada\\
\texttt{shenweih@sfu.ca}
}

\begin{document}
\maketitle
\setcounter{footnote}{0}

\begin{abstract}
A graph is $H$-free if it has no induced subgraph isomorphic to~$H$.
Brandst{\"a}dt, Engelfriet, Le and Lozin proved that the class of chordal graphs with independence number at most~$3$ has unbounded clique-width.
Brandst{\"a}dt, Le and Mosca erroneously claimed that the gem and the co-gem are the only two 1-vertex $P_4$-extensions~$H$
for which the class of $H$-free chordal graphs has bounded clique-width. 
In fact we prove that bull-free chordal and co-chair-free chordal graphs have clique-width at most~$3$ and~$4$, respectively.
In particular, we find four new classes of $H$-free chordal graphs of bounded clique-width.
Our main result, obtained by combining new and known results, provides a classification of all but two stubborn cases, that is, with two potential exceptions we determine {\em all} graphs~$H$ for which the class of $H$-free chordal graphs has bounded clique-width.
We illustrate the usefulness of this classification for classifying other types of graph classes by proving that the class of $(2P_1+\nobreak P_3,\allowbreak K_4)$-free graphs has bounded clique-width via a reduction to
$K_4$-free chordal graphs.
Finally, we give a complete classification of the (un)boundedness of clique-width of $H$-free weakly chordal graphs.
\end{abstract}

\section{Introduction}\label{sec:intro}
\begin{sloppypar}
Clique-width is a well-studied graph parameter; see for example the surveys of Gurski~\cite{Gu07} and Kami\'nski, Lozin and Milani\v{c}~\cite{KLM09}.
In particular, there are numerous graph classes, such as those that can be characterized by
one or more forbidden induced subgraphs,\footnote{See also the Information System on Graph Classes and their Inclusions~\cite{isgci}, which keeps a record of graph classes for which (un)boundedness of clique-width is known.} for which it has been determined
whether or not the class is of {\em bounded clique-width}
(i.e. whether there is a constant~$c$ such that the clique-width of every graph in the class is at most~$c$).
Similar research has been done for variants of clique-width, such as linear clique-width~\cite{HMP12} and power-bounded clique-width~\cite{BGMS14}.
Clique-width is also closely related to other graph width parameters. For instance, it is known that every graph class of bounded treewidth has bounded clique-width but the reverse is not true~\cite{CR05}. Moreover, for any graph class, having bounded clique-width is equivalent to having bounded rank-width~\cite{OS06} and also equivalent to having bounded NLC-width~\cite{Johansson98}.
\end{sloppypar}

Clique-width is a very difficult graph parameter to deal with and our understanding of it is still very limited. 
We do know that computing clique-width is \NP-hard~\cite{FRRS09} but 
we do not know if there exist polynomial-time algorithms for computing the clique-width of even very
restricted graph classes, such as unit interval graphs. Also the problem of
deciding whether a graph has clique-width at most~$c$ for some fixed constant~$c$ is only known to be polynomial-time solvable if
$c\leq 3$~\cite{CHLRR12} and is a long-standing open problem for $c\geq 4$.
Identifying more graph classes of bounded clique-width and determining what kinds of structural properties ensure that a graph class has bounded clique-width increases our understanding of this parameter.
Another important reason for studying these types of questions is that certain classes of \NP-complete problems
become polynomial-time solvable on any graph class~${\cal G}$ of bounded clique-width.\footnote{This follows from results~\cite{CMR00,EGW01,KR03b,Ra07} that assume the existence of a so-called $c$-expression of the input graph $G\in {\cal G}$ 
combined with a result~\cite{Oum08} that such a $c$-expression can be obtained in cubic time for some 
$c\leq 8^{\cw(G)}-1$, where~$\cw(G)$ is the clique-width of the graph~$G$.}
Examples of such problems are
those definable in Monadic Second Order Logic using quantifiers on vertices but not on edges.

In this paper we primarily focus on chordal graphs.
The class of chordal graphs has unbounded clique-width, as it contains the class of proper interval graphs and the class of split graphs, both of which have unbounded
clique-width as shown by Golumbic and Rotics~\cite{GR99b} and Makowsky and Rotics~\cite{MR99}, respectively.
We study the clique-width of subclasses of chordal graphs, but before going into more detail we first give some necessary terminology and notation.

\subsection{Notation}
The {\em disjoint union} $(V(G)\cup V(H), E(G)\cup E(H))$ of two vertex-disjoint graphs~$G$ and~$H$ is denoted by~$G+\nobreak H$ and the disjoint union of~$r$ copies of a graph~$G$ is denoted by~$rG$. The {\em complement} of a graph~$G$, denoted by~$\overline{G}$, has vertex set $V(\overline{G})=\nobreak V(G)$ and an edge between two distinct vertices
if and only if these vertices are not adjacent in~$G$. 
For two graphs~$G$ and~$H$ we write $H\ssi G$ to indicate that~$H$ is an induced subgraph of~$G$.
The graphs $C_r,K_r,K_{1,r-1}$ and~$P_r$ denote the cycle, complete graph, star and path on~$r$ vertices, respectively.
The graph~$S_{h,i,j}$, for $1\leq h\leq i\leq j$, denotes the {\em subdivided claw}, that is
the tree that has only one vertex~$x$ of degree~$3$ and exactly three leaves, which are of distance~$h$,~$i$ and~$j$ from~$x$, respectively.
For a set of graphs $\{H_1,\ldots,H_p\}$, a graph~$G$ is {\em $(H_1,\ldots,H_p)$-free} if it has no induced subgraph isomorphic to a graph in $\{H_1,\ldots,H_p\}$.
A graph~$G$ is {\em chordal} if it is $(C_4,C_5,\ldots)$-free and {\em weakly chordal} if both~$G$ and~$\overline{G}$
are $(C_5,C_6,\ldots)$-free. 
Every chordal graph is weakly chordal.

\subsection{Research Goal and Motivation}
We want to determine all graphs~$H$ for which the class of $H$-free chordal graphs has {\em bounded} clique-width. Our motivation for this research is threefold.

\medskip
\noindent
{\em 1. Generate more graph classes for which a number of \NP-complete problems can be solved in polynomial time.}

\medskip
\noindent
Although many such \NP-complete
problems, such as the {\sc Colouring} problem~\cite{Go80}, are polynomial-time solvable on chordal graphs, many others 
continue to be
\NP-complete for graphs in this class.
To give an example, the well-known {\sc Hamilton Cycle} problem is such a problem.
It is \NP-complete even for strongly chordal split graphs~\cite{Mu96}, but becomes polynomial-time solvable on any graph class of bounded clique-width~\cite{EGW01,Wa94}.
Of course, in order to find new ``islands of tractability'', one may want to consider superclasses of $H$-free chordal graphs instead.
However, already when one considers $H$-free weakly chordal graphs, one does not obtain
any new tractable graph classes.
Indeed, the clique-width of the class of $H$-free graphs is bounded if and only if~$H$ is an induced subgraph of~$P_4$~\cite{DP15}, and as we prove later,
the induced subgraphs of~$P_4$ are also the only graphs~$H$ for which the class of $H$-free weakly chordal graphs has bounded clique-width.
The same classification therefore also follows for superclasses, such as $(H,C_5,C_6,\ldots)$-free
graphs (or $H$-free perfect graphs, to give another example).
Since
forests, or equivalently, 
$(C_3,C_4,\ldots)$-free graphs have bounded clique-width
(see also Lemma~\ref{lem:tree})
it follows that the
class of $(H,C_3,C_4,\dots)$-free graphs has bounded clique-width for every
graph~$H$.
It is therefore a natural question to ask for which graphs~$H$ the
class of $(H,C_4,C_5,\dots)$-free (i.e. $H$-free chordal) graphs has bounded
clique-width.

\medskip
\noindent
{\em 2. Classify the boundedness of the clique-width of $(H_1,H_2)$-free graphs.}

\medskip
\noindent
Classifying the boundedness of clique-width for $H$-free chordal graphs turns
out to be useful for determining the (un)boundedness of the clique-width of
graph classes characterized by two forbidden induced subgraphs~$H_1$ and~$H_2$,
just as the full classification for $H$-free bipartite graphs~\cite{DP14} has
proven to be~\cite{DHP0,DLRR12,DP15}.
To demonstrate this, we will prove that the class of $(2P_1+\nobreak
P_3,\allowbreak K_4)$-free graphs has bounded clique-width via a reduction to
$K_4$-free chordal graphs.
We note that
reducing from a target graph class to another class already known to have bounded clique-width
is an important
technique, which has also been used by others; for instance by Brandst{\"a}dt
et al.~\cite{BELL06} who proved that the class of $(C_4,K_{1,3},4P_1)$-free
graphs has bounded clique-width by reducing these graphs to
$(K_{1,3},4P_1)$-free chordal graphs.
Moreover, in a previous paper~\cite{DHP0} we used it 
for showing the boundedness of the clique-width of three other graph classes of $(H_1,H_2)$-free graphs~\cite{DHP0}.
In that paper we reduced each of these classes to some known subclass of perfect graphs of bounded clique-width
(perfect graphs form a superclass of chordal graphs). In particular, we reduced one of these three classes, namely the class of $(\overline{2P_1+P_2},2P_1+P_3)$-free graphs to
to the class of $\overline{2P_1+P_2}$-free chordal graphs, also known as diamond-free chordal graphs
(the diamond is the graph $\overline{2P_1+P_2}$, see also \figurename~\ref{fig:diamond}), which has bounded clique-width~\cite{GR99b}.

Our new result for the class of $(2P_1+\nobreak P_3,\allowbreak K_4)$-free graphs and the three results of~\cite{DHP0} belong to a line of research, in which
we try to extend
results~\cite{BL02,BELL06,BKM06,BK05,BLM04b,BLM04,BM02,DGP14,DLRR12,MR99} on the clique-width of classes of $(H_1,H_2)$-free graphs in order to try to determine the boundedness or unboundedness of the clique-width of every such graph class~\cite{DHP0,DP15}. 
Including our new result for the case $(2P_1+\nobreak P_3,\allowbreak K_4)$, this
led to a classification of all 
but~$13$ open cases (up to some equivalence relation, see~\cite{DP15}).
\begin{figure}
\begin{center}
\scalebox{0.7}{
{\begin{tikzpicture}[scale=1]
\GraphInit[vstyle=Simple]
\SetVertexSimple[MinSize=6pt]
\Vertex[x=0,y=0]{a}
\Vertex[a=30,d=1]{b}
\Vertex[a=-30,d=1]{e}
\Vertex[x=1.73205080757,y=0]{d}
\Edges(e,a,b,e,d,b)
\end{tikzpicture}}}
\end{center}
\caption{The graph $\overline{2P_1+P_2}$, also known as the diamond.}
\label{fig:diamond}
\end{figure}
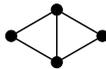

\medskip
\noindent
{\em 3. Complete a line of research on $H$-free chordal graphs.}

\medskip
\noindent
A classification of those graphs~$H$ for which the clique-width of $H$-free chordal graphs is bounded would complete a line of research in the literature, which we feel is
an interesting goal on its own. As a start, using a result of Corneil and Rotics~\cite{CR05} on the relationship between treewidth and clique-width it follows that
the clique-width of the class of $K_r$-free chordal graphs is bounded for all $r\geq 1$.
Brandst{\"a}dt, Engelfriet, Le and Lozin~\cite{BELL06} proved that the class of $4P_1$-free chordal graphs has unbounded clique-width.
Brandst{\"a}dt, Le and Mosca~\cite{BLM04} considered forbidding the graphs $\overline{P_1+P_4}$ (gem) and $P_1+\nobreak P_4$ (co-gem) as induced subgraphs (see also \figurename~\ref{fig:chordal-bounded}). They showed that $(P_1+\nobreak P_4)$-free chordal graphs have clique-width at most~$8$ and also observed that $\overline{P_1+ P_4}$-free chordal graphs belong to the class of distance-hereditary graphs, which have clique-width at most~$3$ (as shown by Golumbic and Rotics~\cite{GR99b}).
Moreover, the same authors~\cite{BLM04} erroneously claimed
that the gem and co-gem are the only two 1-vertex $P_4$-extensions~$H$
for which the class of $H$-free chordal graphs has bounded clique-width.
We prove that bull-free chordal graphs have clique-width at most~$3$, improving a known bound of~$8$, which was shown by Le~\cite{Le03}. We also prove that $\overline{S_{1,1,2}}$-free chordal graphs have clique-width at most~$4$, which Le posed as an open problem.
Results~\cite{BDHP15b,GR99b,MR99} for split graphs and proper interval graphs 
lead to other classes of $H$-free chordal graphs of unbounded clique-width, as we shall discuss
in Section~\ref{sec:prelim}. However, in order to obtain our almost-full dichotomy for $H$-free chordal graphs new results also need to be proved.

\begin{figure}
\begin{center}
\begin{tabular}{cccc}
\begin{minipage}{0.20\textwidth}
\centering
\scalebox{0.7}{
{\begin{tikzpicture}[scale=1,rotate=45]
\GraphInit[vstyle=Simple]
\SetVertexSimple[MinSize=6pt]
\Vertex[x=0,y=0]{a}
\Vertex[a=30,d=1]{b}
\Vertex[a=-30,d=1]{e}
\Vertex[x=1.73205080757,y=0]{d}
\Vertex[x=2.73205080757,y=0]{c}
\Edges(e,a,b,e,d,b)
\Edges(c,d)
\end{tikzpicture}}}
\end{minipage}
&
\begin{minipage}{0.20\textwidth}
\centering
\scalebox{0.7}{
{\begin{tikzpicture}[scale=1,rotate=162]
\GraphInit[vstyle=Simple]
\SetVertexSimple[MinSize=6pt]
\Vertices{circle}{a,b,c,d,e}
\Vertex[a=180,d=1.67504239778]{f}
\Edges(a,b,c,d,e,a)
\Edges(a,c,e,b,d,a)
\Edges(c,f,d)
\end{tikzpicture}}}
\end{minipage}
&
\begin{minipage}{0.20\textwidth}
\centering
\scalebox{0.7}{
{\begin{tikzpicture}[scale=1,rotate=90]
\GraphInit[vstyle=Simple]
\SetVertexSimple[MinSize=6pt]
\Vertices{circle}{a,b,c,d,e}
\Edges(e,a,b,e,d)
\end{tikzpicture}}}
\end{minipage}
&
\begin{minipage}{0.20\textwidth}
\centering
\scalebox{0.7}{
{\begin{tikzpicture}[scale=1,rotate=90]
\GraphInit[vstyle=Simple]
\SetVertexSimple[MinSize=6pt]
\Vertices{circle}{a,b,c,d,e}
\Edges(e,a,b,e,d,b)
\end{tikzpicture}}}
\end{minipage}\\
& & &\\
$\overline{S_{1,1,2}}$ & $\overline{K_{1,3}+2P_1}$ & $P_1+\overline{P_1+P_3}$ & $P_1+\overline{2P_1+P_2}$\\
& & &\\
\begin{minipage}{0.20\textwidth}
\centering
\scalebox{0.7}{
{\begin{tikzpicture}[scale=1,rotate=90]
\GraphInit[vstyle=Simple]
\SetVertexSimple[MinSize=6pt]
\Vertex[x=0,y=0]{a}
\Vertex[a=30,d=1]{b}
\Vertex[a=30,d=2]{c}
\Vertex[a=-30,d=1]{d}
\Vertex[a=-30,d=2]{e}
\Edges(c,b,a,d,e)
\Edges(b,d)
\end{tikzpicture}}}
\end{minipage}
&
\begin{minipage}{0.20\textwidth}
\centering
\scalebox{0.7}{
{\begin{tikzpicture}[scale=1,rotate=90]
\GraphInit[vstyle=Simple]
\SetVertexSimple[MinSize=6pt]
\Vertices{circle}{a,b,c,d,e}
\Edges(a,b,c,d,e,a)
\Edges(a,c,e,b,d,a)
\end{tikzpicture}}}
\end{minipage}
&
\begin{minipage}{0.20\textwidth}
\centering
\scalebox{0.7}{
{\begin{tikzpicture}[scale=1,rotate=90]
\GraphInit[vstyle=Simple]
\SetVertexSimple[MinSize=6pt]
\Vertices{circle}{a,b,c,d,e}
\Edges(b,c,d,e)
\end{tikzpicture}}}
\end{minipage}
&
\begin{minipage}{0.20\textwidth}
\centering
\scalebox{0.7}{
{\begin{tikzpicture}[scale=1,rotate=90]
\GraphInit[vstyle=Simple]
\SetVertexSimple[MinSize=6pt]
\Vertices{circle}{a,b,c,d,e}
\Edges(a,b,c,d,e,a)
\Edges(c,a,d)
\end{tikzpicture}}}
\end{minipage}\\
& & &\\
bull & $K_r$~($r=\nobreak 5$~shown) & $P_1+P_4$ & $\overline{P_1+P_4}$\\
\end{tabular}
\end{center}
\caption{The graphs~$H$ for which the class of $H$-free chordal graphs has bounded clique-width;
the four graphs at the top are new cases proved in this paper.}\label{fig:chordal-bounded}
\end{figure}
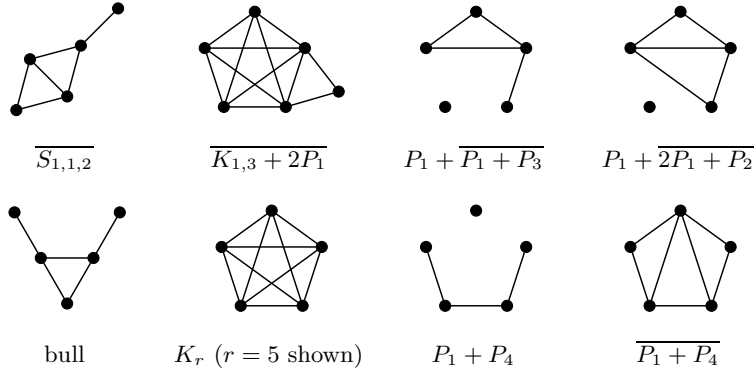

\begin{figure}
\begin{center}
\begin{tabular}{ccc}
\begin{minipage}{0.25\textwidth}
\centering
\scalebox{0.7}{
{\begin{tikzpicture}[scale=1,rotate=135]
\GraphInit[vstyle=Simple]
\SetVertexSimple[MinSize=6pt]
\Vertices{circle}{a,b,c,d}
\Vertex[a=45,d=1.57313218497]{e}
\Vertex[a=225,d=1.57313218497]{f}
\Edges(a,b,c,d,a,c)
\Edges(b,d)
\Edges(a,e)
\Edges(f,d)
\end{tikzpicture}}}
\end{minipage}
&
\begin{minipage}{0.25\textwidth}
\centering
\scalebox{0.7}{
{\begin{tikzpicture}[scale=1,rotate=135]
\GraphInit[vstyle=Simple]
\SetVertexSimple[MinSize=6pt]
\Vertices{circle}{a,b,c,d}
\Vertex[a=45,d=1.57313218497]{e}
\Vertex[a=225,d=1.57313218497]{f}
\Edges(a,b,c,d,a,c)
\Edges(b,d)
\Edges(a,e)
\Edges(c,f,d)
\end{tikzpicture}}}
\end{minipage}\\
&\\
$F_1$ & $F_2$ \\
\end{tabular}
\end{center}
\caption{The graphs~$H$ for which the boundedness of clique-width of the class of $H$-free chordal graphs is open.}
\label{fig:open-chordal}
\end{figure}
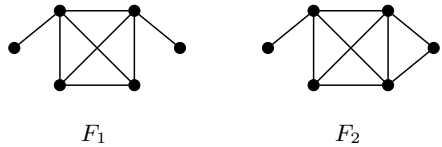

\subsection{Our Results}
In Section~\ref{sec:prelim}, we collect all previously known results for $H$-free chordal graphs and use a result of Olariu~\cite{Olariu91} to prove that bull-free chordal graphs have clique-width at most~$3$.
In Section~\ref{sec:bounded} we present four new classes of $H$-free chordal graphs of bounded
clique-width,\footnote{In Theorems~\ref{thm:co-k13+2p1-chordal},~\ref{thm:p1+paw-chordal} and~\ref{thm:p1+diamond-chordal}, we do not specify our upper bounds as this would
complicate our proofs for negligible gain. In our proofs we repeatedly apply
graph operations that exponentially increase the upper bound on the
clique-width, which means that the bounds that could be obtained from our
proofs would be very large and far from being tight.
Furthermore, we make use of other results that do not give explicit bounds.
We use different
techniques to prove Lemma~\ref{lem:bull-chordal} and Theorem~\ref{thm:co-chair}, and these allow us to give good bounds for these cases.} namely when $H\in \{\overline{K_{1,3}+2P_1},\allowbreak P_1+\nobreak \overline{P_1+P_3},\allowbreak P_1+\nobreak \overline{2P_1+P_2},\allowbreak \overline{S_{1,1,2}}\}$ (see also \figurename~\ref{fig:chordal-bounded}).
In particular, we show that~$\overline{S_{1,1,2}}$-free graphs have clique-width at most~$4$.
One of the algorithmic consequences of these results is that we have identified four new graph classes for which {\sc Hamilton Cycle} is polynomial-time solvable.
In Section~\ref{sec:chordal-classification} we combine all these results with previously known results~\cite{BDHP15b,BELL06,BLM04,GR99b,Le03} to obtain an almost-complete classification for $H$-free chordal graphs (see also \figurename~\ref{fig:chordal-bounded}), leaving only two open cases (see also \figurename~\ref{fig:open-chordal}):
\begin{theorem}\label{thm:chordal-classification}
Let~$H$ be a graph with $H\notin \{F_1,F_2\}$. The class of $H$-free chordal graphs has bounded clique-width if and only if
\begin{itemize}
\item $H=K_r$ for some $r\geq 1$;
\item $H\ssi \bull$;
\item $H\ssi P_1+P_4$;
\item $H\ssi \overline{P_1+P_4}$;
\item $H\ssi \overline{K_{1,3}+2P_1}$;
\item $H\ssi P_1+\overline{P_1+P_3}$;
\item $H\ssi P_1+\overline{2P_1+P_2}$ or
\item $H\ssi \overline{S_{1,1,2}}$.
\end{itemize}
\end{theorem}
In Section~\ref{sec:chordal-classification} 
we also show (using only previously known results) our aforementioned classification for $H$-free weakly chordal graphs.
\begin{theorem}\label{t-weakly-chordal}
Let~$H$ be a graph. The class of
$H$-free weakly chordal graphs has bounded clique-width if and only if~$H$ is an induced subgraph of~$P_4$.
\end{theorem}
In Section~\ref{sec:consec} we illustrate the usefulness of having a classification for $H$-free chordal graphs by proving that the
class of $(2P_1+\nobreak P_3,\allowbreak K_4)$-free graphs has bounded clique-width via a reduction to $K_4$-free chordal graphs.
As such the number of (non-equivalent) pairs $(H_1,H_2)$ for which we do not know whether the clique-width of the class of $(H_1,H_2)$-free graphs is bounded is~13. These remaining cases are given in Section~\ref{sec:conclusions} (see also~\cite{DP15}).
In Section~\ref{sec:conclusions}, we mention a number of future research directions.

\section{Preliminaries}\label{sec:prelim}

All graphs considered in this paper are finite, undirected and have neither multiple edges nor self-loops.
In this section we first define some more standard graph terminology, some additional notation and 
give some structural lemmas. We refer to the textbook of Diestel~\cite{Di12} for any undefined terminology.
Afterwards, we give the definition of clique-width and present a number of known results on clique-width that we will use as lemmas for proving our results.

Let $G=(V,E)$ be a graph.
For $S\subseteq V$, we let~$G[S]$ denote the {\em induced} subgraph of~$G$, which has vertex set~$S$ and edge set $\{uv\; |\; u,v\in S, uv\in E\}$.
If $S=\{s_1,\ldots,s_r\}$ then, to simplify notation, we may also write $G[s_1,\ldots,s_r]$ instead of $G[\{s_1,\ldots,s_r\}]$. For some set $T\subseteq V$ we may write $G-T=G[V\setminus T]$.
Recall that for two graphs~$G$ and~$H$ we write $H\ssi G$ to indicate that~$H$ is an induced subgraph of~$G$.

Let $G=(V,E)$ be a graph.
The set $N(u)=\{v\in V\; |\; uv\in E\}$ is the {\em neighbourhood} of $u\in V$.
The {\em degree} of a vertex $u\in V$ in~$G$ is the size~$|N(u)|$ of its neighbourhood.
The {\em maximum degree} of~$G$ is the maximum vertex degree.
Let $S\subseteq V$.
For a vertex $u\in V$ we write $N_S(u)=N(u) \cap S$.

Let~$S$ and~$T$ be two vertex subsets of a graph $G=(V,E)$ with $S\cap T=\emptyset$.
We say that~$S$ {\em dominates}~$T$ if every vertex of~$T$ is adjacent to at least one vertex of~$S$.
We say that~$S$ is a {\em dominating set} of~$G$ or that~$S$ {\em dominates}~$G$ if every vertex in $V\setminus S$ is adjacent to at least one vertex in~$S$.
We say that~$S$ is {\em complete} to~$T$ if every vertex in~$S$ is adjacent to every vertex in~$T$, and
we say that~$S$ is {\em anti-complete} to~$T$ if every vertex in~$S$ is non-adjacent to every vertex in~$T$.
Similarly, a vertex $v\in V\setminus T$ is {\em complete} or {\em anti-complete} to~$T$ if it is adjacent or non-adjacent, respectively, to every vertex of~$T$.
A set of vertices~$M$ is a {\em module} if every vertex not in~$M$ is either
complete or anti-complete to~$M$. A module in a graph is {\em trivial} if it
contains zero, one or all vertices of the graph, otherwise it is {\em
non-trivial}. We say that~$G$ is {\em prime} if every module in~$G$ is trivial.
We say that a vertex~$v$ {\em distinguishes} two vertices~$x$ and~$y$ if~$v$ is adjacent to precisely one of~$x$ and~$y$. Note that if a set~$M \subseteq V$ is not a module then there must be vertices $x,y \in M$ and a vertex $v \in V \setminus M$ such that~$v$ distinguishes~$x$ and~$y$.

The following two structural lemmas, both of which we need for the proofs of
our results, are about prime graphs containing some specific induced
subgraph~$H$. They
are examples of the well-developed technique of {\em prime extension}, that is, 
they 
show us that such prime graphs must also contain (as an
induced subgraph) at least one of a list of possible extensions of~$H$.
The first prime extension lemma
is due to Brandst\"adt, and the second one is due to Brandst\"adt, Le and de Ridder.

\begin{lemma}[\cite{Brand04}]\label{lem:fig:diamond-prime-ext}
If a prime graph~$G$ contains an induced $\overline{2P_1+P_2}$ then it contains an induced $\overline{P_1+P_4}$, $d$-$\mathbb{A}$ or $d$-domino (see also \figurename~\ref{fig:diamond-prime-ext}).
\end{lemma}

\begin{figure}
\begin{center}
\begin{tabular}{ccc}
\begin{minipage}{0.15\textwidth}
\centering
\scalebox{0.7}{
\begin{tikzpicture}[scale=1,rotate=90]
\GraphInit[vstyle=Simple]
\SetVertexSimple[MinSize=6pt]
\Vertex[x=0,y=0]{x00}
\Vertex[x=1,y=0]{x01}
\Vertex[x=2,y=0]{x02}
\Vertex[x=0,y=1]{x10}
\Vertex[x=1,y=1]{x11}
\Vertex[x=2,y=1]{x12}
\Edges(x00,x01,x02,x12,x11,x10)
\Edges(x01,x11,x02)
\end{tikzpicture}}
\end{minipage}
&
\begin{minipage}{0.15\textwidth}
\centering
\scalebox{0.7}{
{\begin{tikzpicture}[scale=1,rotate=90]
\GraphInit[vstyle=Simple]
\SetVertexSimple[MinSize=6pt]
\Vertices{circle}{a,b,c,d,e}
\Edges(a,b,c,d,e,a)
\Edges(c,a,d)
\end{tikzpicture}}}
\end{minipage}
&
\begin{minipage}{0.15\textwidth}
\centering
\scalebox{0.7}{
\begin{tikzpicture}[scale=1, rotate=90]
\GraphInit[vstyle=Simple]
\SetVertexSimple[MinSize=6pt]
\Vertex[x=0,y=0]{x00}
\Vertex[x=1,y=0]{x01}
\Vertex[x=2,y=0]{x02}
\Vertex[x=0,y=1]{x10}
\Vertex[x=1,y=1]{x11}
\Vertex[x=2,y=1]{x12}
\Edges(x00,x01,x02,x12,x11,x10,x00)
\Edges(x01,x11,x02)
\end{tikzpicture}}
\end{minipage}\\
& &\\
$d$-$\mathbb{A}$ & $\overline{P_1+P_4}$ & $d$-domino
\end{tabular}
\end{center}
\caption{The minimal prime extensions of $\overline{2P_1+P_2}$.}
\label{fig:diamond-prime-ext}
\end{figure}
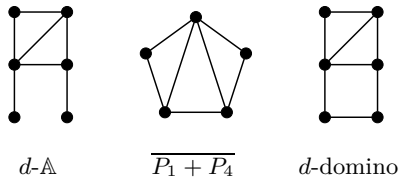

\begin{lemma}[\cite{BLR04}]\label{lem:cogemprimeext}
If a prime graph~$G$ contains an induced subgraph isomorphic to $P_1+\nobreak P_4$ then it contains one of the graphs in \figurename~\ref{fig:co-gem-prime-ext} as an induced subgraph.
\end{lemma}

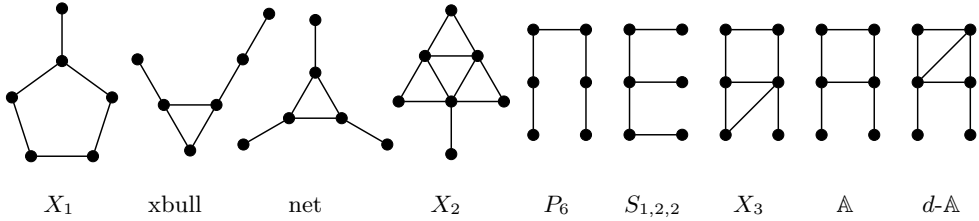
\begin{figure}
\begin{center}
\begin{tabular}{ccccccccc}
\begin{minipage}{0.12\textwidth}
\centering
\scalebox{0.7}{
{\begin{tikzpicture}[scale=1,rotate=90]
\GraphInit[vstyle=Simple]
\SetVertexSimple[MinSize=6pt]
\Vertices{circle}{a,b,c,d,e}
\Vertex[a=0,d=2]{f}
\Edges(a,b,c,d,e,a,f)
\end{tikzpicture}}}
\end{minipage}
&
\begin{minipage}{0.10\textwidth}
\centering
\scalebox{0.7}{
\begin{tikzpicture}[scale=1,rotate=60]
\GraphInit[vstyle=Simple]
\SetVertexSimple[MinSize=6pt]
\Vertex[a=60,d=2]{a}
\Vertex[a=60,d=1]{b}
\Vertex[a=0,d=1]{c}
\Vertex[a=0,d=2]{d}
\Vertex[a=0,d=3]{e}
\Vertex[x=0,y=0]{z}
\Edges(a,b,c,d,e)
\Edges(b,z,c)
\end{tikzpicture}}
\end{minipage}
&
\begin{minipage}{0.15\textwidth}
\centering
\scalebox{0.7}{
{\begin{tikzpicture}[scale=1,rotate=90]
\GraphInit[vstyle=Simple]
\SetVertexSimple[MinSize=6pt]
\Vertex[a=0,d=0.57735026919]{a}
\Vertex[a=120,d=0.57735026919]{b}
\Vertex[a=240,d=0.57735026919]{c}
\Vertex[a=0,d=1.57735026919]{d}
\Vertex[a=120,d=1.57735026919]{e}
\Vertex[a=240,d=1.57735026919]{f}
\Edges(a,b,c,a,d)
\Edges(b,e)
\Edges(c,f)
\end{tikzpicture}}}
\end{minipage}
&
\begin{minipage}{0.12\textwidth}
\centering
\scalebox{0.7}{
{\begin{tikzpicture}[scale=1,rotate=-90]
\GraphInit[vstyle=Simple]
\SetVertexSimple[MinSize=6pt]
\Vertex[a=0,d=0.57735026919]{a}
\Vertex[a=120,d=0.57735026919]{b}
\Vertex[a=240,d=0.57735026919]{c}
\Vertex[a=60,d=1.15470053838]{d}
\Vertex[a=180,d=1.15470053838]{e}
\Vertex[a=300,d=1.15470053838]{f}
\Vertex[a=0,d=1.57735026919]{z}
\Edges(a,b,c,a,d,b,e,c,f,a)
\Edges(a,z)
\end{tikzpicture}}}
\end{minipage}
&
\begin{minipage}{0.09\textwidth}
\centering
\scalebox{0.7}{
\begin{tikzpicture}[scale=1,rotate=90]
\GraphInit[vstyle=Simple]
\SetVertexSimple[MinSize=6pt]
\Vertex[x=0,y=0]{x00}
\Vertex[x=1,y=0]{x01}
\Vertex[x=2,y=0]{x02}
\Vertex[x=0,y=1]{x10}
\Vertex[x=1,y=1]{x11}
\Vertex[x=2,y=1]{x12}
\Edges(x00,x01,x02,x12,x11,x10)
\end{tikzpicture}}
\end{minipage}
&
\begin{minipage}{0.09\textwidth}
\centering
\scalebox{0.7}{
\begin{tikzpicture}[scale=1,rotate=270]
\GraphInit[vstyle=Simple]
\SetVertexSimple[MinSize=6pt]
\Vertex[x=0,y=0]{x00}
\Vertex[x=1,y=0]{x01}
\Vertex[x=2,y=0]{x02}
\Vertex[x=0,y=1]{x10}
\Vertex[x=1,y=1]{x11}
\Vertex[x=2,y=1]{x12}
\Edges(x10,x00,x01,x02,x12)
\Edges(x01,x11)
\end{tikzpicture}}
\end{minipage}
&
\begin{minipage}{0.09\textwidth}
\centering
\scalebox{0.7}{
\begin{tikzpicture}[scale=1,rotate=90]
\GraphInit[vstyle=Simple]
\SetVertexSimple[MinSize=6pt]
\Vertex[x=0,y=0]{x00}
\Vertex[x=1,y=0]{x01}
\Vertex[x=2,y=0]{x02}
\Vertex[x=0,y=1]{x10}
\Vertex[x=1,y=1]{x11}
\Vertex[x=2,y=1]{x12}
\Edges(x00,x01,x02,x12,x11,x10,x01)
\Edge(x01)(x11)
\end{tikzpicture}}
\end{minipage}
&
\begin{minipage}{0.09\textwidth}
\centering
\scalebox{0.7}{
\begin{tikzpicture}[scale=1,rotate=90]
\GraphInit[vstyle=Simple]
\SetVertexSimple[MinSize=6pt]
\Vertex[x=0,y=0]{x00}
\Vertex[x=1,y=0]{x01}
\Vertex[x=2,y=0]{x02}
\Vertex[x=0,y=1]{x10}
\Vertex[x=1,y=1]{x11}
\Vertex[x=2,y=1]{x12}
\Edges(x00,x01,x02,x12,x11,x10)
\Edge(x01)(x11)
\end{tikzpicture}}
\end{minipage}
&
\begin{minipage}{0.09\textwidth}
\centering
\scalebox{0.7}{
\begin{tikzpicture}[scale=1,rotate=90]
\GraphInit[vstyle=Simple]
\SetVertexSimple[MinSize=6pt]
\Vertex[x=0,y=0]{x00}
\Vertex[x=1,y=0]{x01}
\Vertex[x=2,y=0]{x02}
\Vertex[x=0,y=1]{x10}
\Vertex[x=1,y=1]{x11}
\Vertex[x=2,y=1]{x12}
\Edges(x00,x01,x02,x12,x11,x10)
\Edges(x01,x11,x02)
\end{tikzpicture}}
\end{minipage}\\
& & & & & & & &\\
$X_1$ & xbull & $\net$ & $X_2$ & $P_6$ & $S_{1,2,2}$ & $X_3$ & $\mathbb{A}$ & $d$-$\mathbb{A}$
\end{tabular}
\end{center}
\caption{The minimal prime extensions of $P_1+\nobreak P_4$.}
\label{fig:co-gem-prime-ext}
\end{figure}

We also use the following structural lemma due to Olariu.

\begin{lemma}[\cite{Olariu91}]\label{lem:bull-house-prime}
Every prime $(\bull,\house)$-free graph (see also \figurename~\ref{fig:bull-and-house}) is either $K_3$-free or the complement of a $2P_2$-free bipartite graph.
\end{lemma}

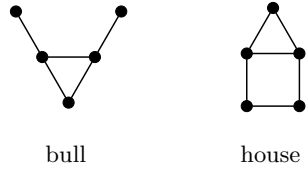
\begin{figure}
\begin{center}
\begin{tabular}{cc}
\begin{minipage}{0.20\textwidth}
\centering
\scalebox{0.7}{
{\begin{tikzpicture}[scale=1,rotate=90]
\GraphInit[vstyle=Simple]
\SetVertexSimple[MinSize=6pt]
\Vertex[x=0,y=0]{a}
\Vertex[a=30,d=1]{b}
\Vertex[a=30,d=2]{c}
\Vertex[a=-30,d=1]{d}
\Vertex[a=-30,d=2]{e}
\Edges(c,b,a,d,e)
\Edges(b,d)
\end{tikzpicture}}}
\end{minipage}
&
\begin{minipage}{0.2\textwidth}
\centering
\scalebox{0.7}{
\begin{tikzpicture}[scale=1]
\GraphInit[vstyle=Simple]
\SetVertexSimple[MinSize=6pt]
\Vertex[x=0,y=0]{x00}
\Vertex[x=1,y=0]{x01}
\Vertex[x=0,y=1]{x10}
\Vertex[x=1,y=1]{x11}
\Vertex[x=0.5,y=1.86602540378]{t}
\Edges(x00,x01,x11,x10,x00)
\Edges(x11,t,x10)
\end{tikzpicture}}
\end{minipage}\\
\\
bull & house 
\end{tabular}
\end{center}
\caption{The graphs bull and house.}\label{fig:bull-and-house}
\end{figure}

Let $G=(V,E)$ be a connected graph.
An edge $e \in E$ is a {\em bridge} if deleting it would make~$G$ disconnected.
A vertex $v\in V$ is a {\em cut-vertex} if $G[V\setminus \{v\}]$ is disconnected.
If~$G$ has at least three vertices, but no cut-vertices then it is {\em $2$-connected}.
For any two vertices~$u$ and~$v$ in a 2-connected graph, there are two paths from~$u$ to~$v$ that are internally vertex-disjoint (by Menger's Theorem, see e.g.~\cite{Di12}).
A {\em block} of~$G$ is a maximal 2-connected subgraph, a bridge or a single vertex.
Note that two blocks of~$G$ have at most one common vertex, which must be a cut-vertex of~$G$.

Recall that~$K_{1,r}$ denotes the $(r+\nobreak 1)$-vertex star.
In this graph the vertex of degree~$r$ is called the {\em central vertex}.
A {\em double-star} is the graph formed from
two stars~$K_{1,s}$ and~$K_{1,r}$ by joining the central vertices of each star
with an edge.

Let $G=(V,E)$ be a graph.
A set $S\subseteq V$ is {\em independent} if~$G[S]$ contains no edges.
The {\em independence number} of~$G$ is the size of a largest independent set of~$G$.
If~$V$ can be partitioned into two (possibly empty) independent sets then~$G$ is {\em bipartite}.
We say that~$G$ is \emph{complete multipartite} if~$V$ can be partitioned into~$k$ independent sets $V_1,\ldots,V_k$ (called {\em partition classes}) for some integer~$k$, such that two vertices are adjacent if and only if they belong to two different sets~$V_i$ and~$V_j$.

The next result, which we will use later on, is due to Olariu~\cite{Olariu88} (note that the graph $\overline{P_1+P_3}$ is also called the {\em paw}).

\begin{lemma}[\cite{Olariu88}]\label{lem:paw}
Every connected $(\overline{P_1+P_3})$-free graph is either complete multipartite or $K_3$-free.
\end{lemma}

Let $G=(V,E)$ be a graph. A vertex $v\in V$ is {\em simplicial} if~$G[N(v)]$ is complete.
The following lemma is well known (see e.g.~\cite{Go80}).

\begin{lemma}\label{l-simplicial}
Every chordal graph has a simplicial vertex.
\end{lemma}

Let $G=(V,E)$ be a graph.
A set $S\subseteq V$ is said to be a {\em clique} if~$G[S]$ is a complete graph.
The {\em clique number} of~$G$ is the size of a largest clique of~$G$.
The {\em chromatic number} of~$G$ is the minimum number~$k$ for which~$G$ has a {\em $k$-colouring}, that is, for which there exists a mapping $c:V\to \{1,\ldots,k\}$ such
that $c(u)\neq c(v)$ whenever~$u$ and~$v$ are adjacent.
We say that~$G$ is {\em perfect} if, for every induced subgraph $H\ssi G$, the chromatic number of~$H$ equals its clique number.
The graph~$G$ is a {\em split graph} if it has a {\em split partition}, that is, a partition of~$V$ into two (possibly empty) sets~$K$ and~$I$, where~$K$ is a clique and~$I$ is an independent set; if~$K$ and~$I$ are complete to each other then~$G$ is said to be a {\em complete} split graph.

It is well known that every split graph is chordal
and that every chordal graph is perfect (see~\cite{Go80}).
The first inclusion also follows from the next lemma, which is due to F\"oldes and Hammer~\cite{FH77}.
\begin{lemma}[\cite{FH77}]\label{lem:split}
A graph is split if and only if it is $(C_4,C_5,2P_2)$-free.
\end{lemma}

\phantomsection\label{def:thin-spider}
A graph~$G$ is a {\em thin spider} if its vertex set can be partitioned into a
clique~$K$, an independent set~$I$ and a set~$R$ such that $|K|=|I| \geq 2$,
the set~$R$ is complete to~$K$ and anti-complete to~$I$ and the edges
between~$K$ and~$I$ form an induced matching (that is, every vertex of~$K$ has a
unique neighbour in~$I$ and vice versa). Note that if a thin spider is prime
then~$|R| \leq 1$. A {\em thick spider} is the complement of a thin spider. A graph is a
{\em spider} if it is either a thin or a thick spider.

Spiders play an important role in our result for $\overline{S_{1,1,2}}$-free chordal graphs and we will need the following lemma (due to Brandst\"adt and Mosca).

\begin{lemma}[\cite{BM04}]\label{lem:chair-split-spider}
If~$G$ is a prime $S_{1,1,2}$-free split graph then it is a spider.
\end{lemma}

\subsection{Clique-width}

The {\em clique-width} of a graph~$G$, denoted by~$\cw(G)$, is the minimum
number of labels needed to
construct~$G$ by
using the following four operations:
\begin{enumerate}
\item creating a new graph consisting of a single vertex~$v$ with label~$i$ (denoted by~$i(v)$);
\item taking the disjoint union of two labelled graphs~$G_1$ and~$G_2$ (denoted by $G_1\oplus\nobreak G_2$);
\item joining each vertex with label~$i$ to each vertex with label~$j$ ($i\neq j$, denoted by~$\eta_{i,j}$);
\item renaming label~$i$ to~$j$ (denoted by~$\rho_{i\rightarrow j}$).
\end{enumerate}
An algebraic term that represents such a construction of~$G$ and uses at most~$k$ labels is said to be a {\em $k$-expression} of~$G$ (i.e. the clique-width of~$G$ is the minimum~$k$ for which~$G$ has a $k$-expression).
For instance, an induced path on four consecutive vertices $a, b, c, d$ has clique-width equal to~$3$, and the following 3-expression can be used to construct~it:
$$
\eta_{3,2}(3(d)\oplus \rho_{3\rightarrow 2}(\rho_{2\rightarrow 1}(\eta_{3,2}(3(c)\oplus \eta_{2,1}(2(b)\oplus 1(a)))))).
$$
A class of graphs~${\cal G}$ has \emph{bounded} clique-width if
there is a constant~$c$ such that the clique-width of every graph in~${\cal G}$ is at most~$c$; otherwise the clique-width of~${\cal G}$ is
{\em unbounded}.

Let~$G$ be a graph. We define the following operations.
The \emph{subdivision} of an edge~$uv$ replaces~$uv$ by a new vertex~$w$ with edges~$uw$ and~$vw$.
For an induced subgraph $G'\ssi G$, the {\em subgraph complementation} operation (acting on~$G$ with respect to~$G'$) replaces every edge present in~$G'$
by a non-edge, and vice versa. Similarly, for two disjoint vertex subsets~$S$ and~$T$ in~$G$, the {\em bipartite complementation} operation with respect to~$S$ and~$T$ acts on~$G$ by replacing
every edge with one end-vertex in~$S$ and the other one in~$T$ by a non-edge and vice versa.

We now state some useful facts about how the above operations (and some other ones) influence the clique-width of a graph.
We will use these facts throughout the paper.
Let $k\geq 0$ be a constant and let~$\gamma$ be some graph operation.
We say that a graph class~${\cal G'}$ is {\em $(k,\gamma)$-obtained} from a graph class~${\cal G}$
if the following two conditions hold:
\begin{enumeratei}
\item every graph in~${\cal G'}$ is obtained from a graph in~${\cal G}$ by performing~$\gamma$ at most~$k$ times, and
\item for every $G\in {\cal G}$ there exists at least one graph
in~${\cal G'}$ obtained from~$G$ by performing~$\gamma$ at most~$k$ times.
\end{enumeratei}
If we do not impose a finite upper bound~$k$ on the number of applications of~$\gamma$ then we write that~${\cal G}'$ is $(\infty,\gamma)$-obtained from~${\cal G}$.

We say that~$\gamma$ {\em preserves} boundedness of clique-width if
for any finite constant~$k$ and any graph class~${\cal G}$, any graph class~${\cal G}'$ that is $(k,\gamma)$-obtained from~${\cal G}$
has bounded clique-width if and only if~${\cal G}$ has bounded clique-width.

\begin{enumerate}[\bf F{a}ct 1.]
\item \label{fact:del-vert} Vertex deletion preserves boundedness of clique-width~\cite{LR04}.\\[-1em]

\item \label{fact:comp} Subgraph complementation preserves boundedness of clique-width~\cite{KLM09}.\\[-1em]

\item \label{fact:bip} Bipartite complementation preserves boundedness of clique-width~\cite{KLM09}.\\[-1em]

\item \label{fact:2-conn} If~${\cal G}$ is a class of graphs, then~${\cal G}$ has
bounded clique-width if and only if the class of 2-connected induced subgraphs of graphs in~${\cal G}$
has bounded clique-width~\cite{BL02,LR04}.

\item \label{fact:subdiv}
For a class of graphs~${\cal G}$ of {\em bounded} maximum
degree, let~${\cal G}'$ be a class of graphs that is $(\infty,{\tt es})$-obtained from~${\cal G}$, where~${\tt es}$ is the edge subdivision operation.
Then~${\cal G}$ has bounded clique-width if and only if~${\cal G}'$ has
bounded clique-width~\cite{KLM09}.
\end{enumerate}
We also use a number of other elementary results on the clique-width of graphs.
The first two are well known and straightforward to check.

\begin{lemma}\label{lem:tree}
The clique-width of a forest is at most~$3$.
\end{lemma}

\begin{lemma}\label{lem:atmost-2}
The clique-width of a graph with maximum degree at most~$2$ is at most~$4$.
\end{lemma}

The following lemma tells us that if~${\cal G}$ is a hereditary graph class (i.e. a graph class closed under vertex deletion) then 
in order to determine whether~${\cal G}$ has bounded clique-width
we may restrict ourselves to the graphs in~${\cal G}$ that are prime.

\begin{lemma}[\cite{CO00}]\label{lem:prime}
Let $G$ be a graph and let~${\cal P}$ be the set of all induced subgraphs of~$G$ that are prime.
Then $\cw(G)=\max_{H \in {\cal P}}\cw(H)$.
\end{lemma}

\subsection{Known Results on $H$-free Chordal Graphs}\label{subsec:known-results}

To prove our results,
we need to use a number of known results. We present these results as lemmas in this subsection.
The first of these lemmas gives a classification for $H$-free graphs.

\begin{lemma}[\cite{DP15}]\label{l-p4}
Let~$H$ be a graph. The class of
$H$-free graphs has bounded clique-width if and only if~$H$ is an induced subgraph of~$P_4$.
\end{lemma}

We will use the following characterization of graphs~$H$ for which the class of $H$-free bipartite graphs has bounded clique-width
(which is similar to a characterization of Lozin and Volz~\cite{LV08} for a different variant of the notion of $H$-freeness in bipartite graphs, see~\cite{DP14} for an explanation of the difference).

\begin{lemma}[\cite{DP14}]\label{lem:bipartite}
Let~$H$ be a graph. The class of $H$-free bipartite graphs has bounded
clique-width if and only if one of the following cases holds:
\begin{itemize}
\item $H=sP_1$ for some $s\geq 1$;
\item $H\ssi K_{1,3}+3P_1$;
\item $H\ssi K_{1,3}+P_2$;
\item $H\ssi P_1+S_{1,1,3}$;
\item $H\ssi S_{1,2,3}$.
\end{itemize}
\end{lemma}

For a graph~$G$, let~$\tw(G)$ denote the treewidth of~$G$ (see, for example, Diestel~\cite{Di12} for a definition of this notion).
Corneil and Rotics~\cite{CR05} showed that $\cw(G) \leq 3\times 2^{\tw(G)-1}$ for every graph~$G$.
Because the treewidth of a chordal graph is equal to the size of a maximum clique minus~$1$ (see e.g.~\cite{Bo93}), this result leads to the following well-known lemma.

\begin{lemma}\label{lem:clique-chordal}
The class of $K_r$-free chordal graphs has bounded clique-width for all $r\geq\nobreak 1$.
\end{lemma}

The {\em bull} is the graph obtained from the cycle~$abca$ after adding two new vertices~$d$ and~$e$ with edges $ad,be$ (see also \figurename~\ref{fig:chordal-bounded}).
In~\cite{BLM04}, Brandst{\"a}dt, Le and Mosca erroneously claimed that the
clique-width of $\overline{S_{1,1,2}}$-free chordal graphs and of bull-free chordal graphs is
unbounded. 
Using a general result of De Simone~\cite{DeSimone93}, 
Le~\cite{Le03} proved that every bull-free chordal graph has clique-width at most~$8$.
Using a result of Olariu~\cite{Olariu91} we can prove the following

\begin{lemma}\label{lem:bull-chordal}
Every bull-free chordal graph has clique-width at most~$3$.
\end{lemma}

\begin{proof}
Let~$G$ be a bull-free chordal graph.
By Lemma~\ref{lem:prime}, we may assume that~$G$ is prime.
Note that the house contains an induced~$C_4$, so~$G$ is house-free.
Then, by Lemma~\ref{lem:bull-house-prime},~$G$ is either $K_3$-free or the complement of a $2P_2$-free bipartite graph.
Every $K_3$-free chordal graph is a forest, so by Lemma~\ref{lem:tree} it has clique-width at
most~$3$. We may therefore assume that~$G$ is a prime graph that is the complement of a $2P_2$-free bipartite graph.
Such graphs are known as $k$-webs in~\cite{Olariu91}, where $k \geq 2$.
A $k$-web consists of two cliques $X=\{x_1,\ldots,x_k\}$ and $Y=\{y_1,\ldots,y_k\}$
such that for $i,j \in \{1,\ldots,k\}$ the vertex~$x_i$ is adjacent to~$y_j$ if
and only if $i<j$.
We will show how to use the operations of clique-width constructions to inductively build a copy of a $k$-web in which every vertex in the set~$X$ is labelled~$1$ and every vertex in the set~$Y$ is labelled~$2$.
Consider a $k$-web labelled as described above for some $k \geq 0$
(if $k=0$ this is the empty graph).
Add a vertex labelled~$3$ to the graph, join it to every vertex of label~$1$ and to every vertex of label~$2$, then relabel it to have label~$1$.
Next, add a vertex labelled~$3$ to the graph, join it to every vertex of label~$2$, then relabel it to have label~$2$.
This is precisely the $(k+1)$-web, also labelled as described above.
We conclude that every $k$-web can be constructed using at most~$3$ labels, so~$G$ has clique-width at most~$3$.\qed
\end{proof}

Since~$P_4$ is a bull-free chordal graph and has clique-width~$3$, the bound in the above lemma is tight.

Next we recall the aforementioned results of Brandst{\"a}dt et al.

\begin{lemma}[\cite{BLM04}]\label{lem:cogem-chordal}
Every $P_1+\nobreak P_4$-free chordal graph has clique-width at most~$8$.
\end{lemma}

\begin{lemma}[\cite{BLM04}]\label{lem:gem-chordal}
Every $\overline{P_1+P_4}$-free chordal graph has clique-width at most~$3$.
\end{lemma}

\begin{lemma}[\cite{BELL06}]\label{lem:4p1-chordal}
The class of $4P_1$-free chordal graphs has unbounded clique-width.
\end{lemma}

Recall that Golumbic and Rotics~\cite{GR99b} proved that the class of proper interval graphs has unbounded clique-width.
Such graphs are well-known to be $K_{1,3}$-free and
chordal~\cite{Ro69}.

\begin{lemma}\label{lem:claw-chordal}
The class of $K_{1,3}$-free chordal graphs has unbounded clique-width.
\end{lemma}

The next lemma is obtained from combining Lemma~\ref{lem:split}
with the aforementioned result of Makowsky and Rotics~\cite{MR99}, who showed that the class of split graphs has unbounded clique-width.

\begin{lemma}\label{lem:split-chordal}
The class of $(C_4,C_5,2P_2)$-free graphs (or equivalently split graphs) has
unbounded clique-width.
\end{lemma}

We note that Lemma~\ref{lem:split-chordal} also follows from a result of Korpelainen, Lozin and Mayhill~\cite{KLM14},
who proved that the class of split permutation graphs has unbounded clique-width.
Moreover, Lemma~\ref{lem:split-chordal} implies that the class of $H$-free chordal graphs has unbounded clique-width for $H\in \{C_4,C_5,2P_2\}$.

\begin{figure}
\begin{center}
\begin{tabular}{ccccccc}
\begin{minipage}{0.15\textwidth}
\centering
\scalebox{0.7}{
{\begin{tikzpicture}[scale=1]
\GraphInit[vstyle=Simple]
\SetVertexSimple[MinSize=6pt]
\Vertex[x=0,y=1]{a}
\Vertex[x=0,y=0]{b}
\Vertex[x=0.5,y=0]{c}
\Vertex[x=-0.5,y=0]{d}
\Edges(c,a,b)
\Edges(a,d)
\end{tikzpicture}}}
\end{minipage}
&
\begin{minipage}{0.15\textwidth}
\centering
\scalebox{0.7}{
{\begin{tikzpicture}[scale=1,rotate=45]
\GraphInit[vstyle=Simple]
\SetVertexSimple[MinSize=6pt]
\Vertices{circle}{a,b,c,d}
\end{tikzpicture}}}
\end{minipage}\\
& \\
$K_{1,3}$ & $4P_1$\\
\end{tabular}
\end{center}
\caption{Graphs~$H$ for which the class of $H$-free chordal graphs was previously known to have unbounded clique-width.}\label{fig:old-unbounded}
\end{figure}
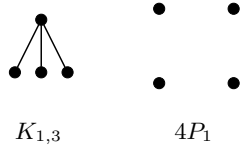

Recall that by Lemma~\ref{lem:split}, every split graph is a chordal graph.
Therefore, if the class of $H$-free chordal graphs has bounded clique-width then the class of $H$-free split graphs must also have bounded clique-width.
To prove Theorem~\ref{thm:chordal-classification}, we will make heavy use of the following lemma.
This lemma can be seen as a refinement of Lemma~\ref{lem:split-chordal}, as it
classifies all but two graphs~$H$ (up to complementation) for which the class
of $H$-free split graphs has bounded clique-width.
(Note that a graph is a split graph if and only if its complement is a split
graph, so by Fact~\ref{fact:comp}, the class of $H$-free split graphs has
bounded clique-width if and only if the class of $\overline{H}$-free split
graphs has bounded clique-width.)

\begin{lemma}[\cite{BDHP15b}]\label{lem:split-classification}
Let~$H$ be a graph such that neither~$H$ nor~$\overline{H}$ is in $\{F_4,F_5\}$ (see \figurename~\ref{fig:open-split}).
The class of $H$-free split graphs has bounded clique-width if and only if
\\[-15pt]
\begin{itemize}
\item $H$ or~$\overline{H}$ is isomorphic to~$rP_1$ for some $r \geq 1$;
\item $H$ or $\overline{H} \ssi F_4$ or
\item $H$ or $\overline{H} \ssi F_5$.
\end{itemize}
\end{lemma}
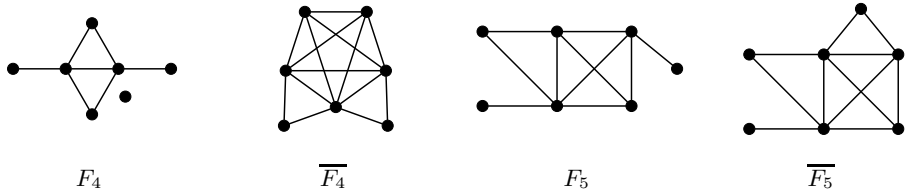
\begin{figure}[h!]
\begin{center}
\begin{tabular}{cccc}
\begin{minipage}{0.24\textwidth}
\centering
\scalebox{0.7}{
{\begin{tikzpicture}[scale=1]
\GraphInit[vstyle=Simple]
\SetVertexSimple[MinSize=6pt]
\Vertex[x=0,y=0]{a}
\Vertex[a=60,d=1]{b}
\Vertex[a=-60,d=1]{c}
\Vertex[a=0,d=1]{d}
\Vertex[a=0,d=2]{x}
\Vertex[a=0,d=-1]{y}
\Vertex[a=-25,d=1.25]{z}
\Edges(y,a,b,d,c,a,d,x)
\end{tikzpicture}}}
\end{minipage}
&
\begin{minipage}{0.24\textwidth}
\centering
\scalebox{0.7}{
{\begin{tikzpicture}[scale=1,rotate=54]
\GraphInit[vstyle=Simple]
\SetVertexSimple[MinSize=6pt]
\Vertices{circle}{a,b,c,d,e}
\Vertex[a=180,d=1.67504239778]{f}
\Vertex[a=252,d=1.67504239778]{g}
\Edges(a,b,c,d,e,a)
\Edges(a,c,e,b,d,a)
\Edges(c,f,d)
\Edges(d,g,e)
\end{tikzpicture}}}
\end{minipage}
&
\begin{minipage}{0.24\textwidth}
\centering
\scalebox{0.7}{
{\begin{tikzpicture}[scale=1,rotate=135]
\GraphInit[vstyle=Simple]
\SetVertexSimple[MinSize=6pt]
\Vertices{circle}{a,b,c,d}
\Vertex[a=225,d=1.57313218497]{f}
\Vertex[x=1,y=2]{y}
\Vertex[x=2,y=1]{z}
\Edges(a,b,c,d,a,c)
\Edges(b,d)
\Edges(a,z,b,y)
\Edges(f,d)
\end{tikzpicture}}}
\end{minipage}
&
\begin{minipage}{0.24\textwidth}
\centering
\scalebox{0.7}{
{\begin{tikzpicture}[scale=1,rotate=135]
\GraphInit[vstyle=Simple]
\SetVertexSimple[MinSize=6pt]
\Vertices{circle}{a,b,c,d}
\Vertex[a=315,d=1.57313218497]{f}
\Vertex[x=1,y=2]{y}
\Vertex[x=2,y=1]{z}
\Edges(a,b,c,d,a,c)
\Edges(b,d)
\Edges(a,z,b,y)
\Edges(a,f,d)
\end{tikzpicture}}}
\end{minipage}\\
& \\
$F_4$ & $\overline{F_4}$ & $F_5$ & $\overline{F_5}$
\end{tabular}
\end{center}
\caption{\label{fig:open-split} The four graphs for which it is not known whether or not the class of $H$-free split graphs has bounded clique-width.
(Recall that the cases~$F_4$ and~$\overline{F_4}$ are equivalent to each-other and that the cases~$F_5$ and~$\overline{F_5}$ are also equivalent to each-other.)}
\end{figure}

A graph $G=(V,E)$ is a {\em permutation graph} if each vertex $v\in V$ corresponds to a straight line segment~$l_v$ between two parallel lines such that there is an edge between two vertices~$u$ and~$v$ if and only if~$l_u$ and~$l_v$ intersect.
A graph is {\em bipartite permutation} if it is both bipartite and permutation.
We need the following result due to Brandst\"adt and Lozin, which we will use in the proof of Theorem~\ref{t-weakly-chordal}.

\begin{lemma}[\cite{BL03}]\label{l-bippermut}
The class of bipartite permutation graphs has unbounded clique-width.
\end{lemma}

\section{New Classes of Bounded Clique-width}\label{sec:bounded}

We present four new classes of $H$-free chordal graphs that have bounded clique-width, namely when $H\in \{ \overline{K_{1,3}+2P_1},\allowbreak P_1+\nobreak \overline{P_1+P_3},\allowbreak P_1+\nobreak \overline{2P_1+P_2},\allowbreak \overline{S_{1,1,2}}\}$.
We prove that these classes have bounded clique-width in the subsections below, making use of known results from
Section~\ref{sec:prelim}.
In particular we will often use Facts~\ref{fact:del-vert}--\ref{fact:subdiv}. Note that Facts~\ref{fact:del-vert} and~\ref{fact:2-conn} can be used safely, since every class of
$H$-free chordal graphs is closed under vertex deletion (when applying the other three facts we need to be more careful).

\subsection{The Case $H=\overline{K_{1,3}+2P_1}$}

Here is our first result.
To prove it, we make use of the celebrated Menger's Theorem (see e.g.~\cite{Di12}) and the facts from Section~\ref{sec:prelim}. In particular Fact~\ref{fact:2-conn}, which states
that a graph~$G$ has bounded clique-width if and only if every block of~$G$ has bounded clique-width, will play an important role in our proof.

\begin{theorem}\label{thm:co-k13+2p1-chordal}
The class of $\overline{K_{1,3}+2P_1}$-free chordal graphs has bounded clique-width.
\end{theorem}

\begin{proof}
Let~$G$ be a $\overline{K_{1,3}+2P_1}$-free chordal graph.
By Fact~\ref{fact:2-conn} we may assume that~$G$ is 2-connected.
Let~$K$ be a maximum clique in~$G$ on~$k$ vertices.
We may assume that
$k \geq 7$, otherwise~$G$ is $K_7$-free, in which case~$G$ has bounded clique-width by Lemma~\ref{lem:clique-chordal}.
We let~$S$ be the set of vertices outside~$K$ with at least two neighbours in~$K$.
Because~$K$ is maximum, $k\geq 7$
and~$G$ is $\overline{K_{1,3}+2P_1}$-free,
every vertex in~$S$ has either one non-neighbour or two non-neighbours in~$K$.

We will prove that $V(G)=K\cup S$.
To this end, we first prove that~$G-S$ is connected.
Suppose, for contradiction, that there exists a vertex~$x$
that is in a connected component~$D$ of~$G-S$
other than the component containing~$K$. Let $u\in K$.
Because~$G$ is 2-connected, it contains two paths~$P_1$
and~$P_2$ from~$x$ to~$u$ that are internally vertex-disjoint (by Menger's Theorem).
Note that we may assume that each~$P_i$ is induced.
For $i=1,2$, let $s_i\in P_i$ be the first vertex that is not in~$D$ and let~$x_i$ be the predecessor of~$s_i$ on~$P_i$. Note that $s_1,s_2 \in S$. Since $k
\geq 7$, there must be a vertex $u' \in K$ adjacent to both~$s_1$ and~$s_2$.
For $i=1,2$ let~$P_i'$ be the path from~$x$ to~$u'$ formed by taking the part of the
path~$P_i$ from~$x$ to~$s_i$ and adding~$u'$. Note that~$P_1'$ and~$P_2'$ are
both induced paths in~$G$ and each contains exactly one vertex from~$K$ and one
from~$S$.
Since~$G$ is chordal,~$s_1$ and~$s_2$ must be adjacent and at least one of~$x_1$ and~$x_2$ must be adjacent to both~$s_1$ and~$s_2$.
Without loss of generality, we assume that~$x_1$ is adjacent to both~$s_1$ and~$s_2$.
On the other hand,~$s_1$ and~$s_2$ have at least three common neighbours in~$K$ since $k\ge 7$.
Then $x_1,s_1,s_2$, together with these three vertices in~$K$ form an induced
$\overline{K_{1,3}+2P_1}$, a contradiction. Thus~$G-S$ is indeed connected.

Suppose, for contradiction, that $V(G) \neq K\cup S$. Since~$G-S$ is connected, there must be a vertex $y \in V(G) \setminus (K \cup S)$ adjacent to a vertex $v \in K$.
As $y\notin S$,~$y$ is anti-complete to $K \setminus \{v\}$.
Let $u \in K \setminus \{v\}$. Since~$G$ is 2-connected, there must exist
an induced path~$P$ from~$y$ to~$u$ with
$v\notin V(P)$ (by Menger's Theorem).
Then~$v$ is complete
to~$V(P)$ since~$G$ is chordal. Let~$y'$ be the last vertex (from~$y$ to~$u$)
on~$P$ that is not in $K\cup S$
(note that~$y'$ is not necessarily distinct from~$y$).
Let~$s$ be the successor of~$y'$ on~$P$.
Since $y'\notin S$ and~$y'$ is adjacent to~$v$, we find that~$y'$ is anti-complete to
$K\setminus \{v\}$. Hence, $s \notin K$, so $s\in S$.
Moreover,~$s$ and~$v$ have at least three common neighbours
in $K\setminus \{v\}$ since $k\ge 7$.
Then~$y',s,v$, together with these three vertices in~$K$
form an induced $\overline{K_{1,3}+2P_1}$, a contradiction.

\medskip
\noindent
For $i=1,2$, let~$S_i$ consist of those vertices with exactly~$i$ non-neighbours in~$K$.
Because every vertex in~$S$ has either one non-neighbour or two non-neighbours, we find that $S=S_1\cup S_2$.

We will now prove, via Claims~\ref{clm:co-k13+2p1-chordal:s2-same-k-non-nhbrs}--\ref{clm:co-k13+2p1-chordal:s-triangle-free}, that~$G[S]$ is a forest.

\clm{\label{clm:co-k13+2p1-chordal:s2-same-k-non-nhbrs}
Any two adjacent vertices in~$S_2$ have the same pair of non-neighbours in~$K$.}
This follows directly from the fact that~$G$ is chordal.

\clm{\label{clm:co-k13+2p1-chordal:s2-common-non-nhbr}
Any two non-adjacent vertices in~$S_2$ have a common non-neighbour.}
Suppose that this is not the case. Then there exist two non-adjacent vertices
$t,t' \in S_2$ and four distinct vertices $a,b,c,d \in K$ with~$t$ non-adjacent to~$a$ and~$b$ and with~$t'$ non-adjacent to~$c$ and~$d$.
As~$t$ and~$t'$ belong to~$S_2$, it follows that~$t$ is adjacent
to~$c$ and~$d$, and that~$t'$ is adjacent to~$a$ and~$b$. Since $k\geq 7$,
we find that~$t$ and~$t'$ have two common neighbours in~$K$.
These two common neighbours, together with $c,d,t,t'$ form an induced $\overline{K_{1,3}+2P_1}$,
a contradiction.

\clm{\label{clm:co-k13+2p1-chordal:s1-s2-common-non-nhbr}
If a vertex $s \in S_1$ is adjacent to a vertex $t \in S_2$ then~$s$ and~$t$ must have a common non-neighbour in~$K$.}
Indeed, let~$v$ be the unique non-neighbour
of~$s$ in~$K$. Then~$v$ must be a non-neighbour of~$t$ otherwise a non-neighbour of~$t$ in~$K$,
together with~$s,t$ and~$v$ would induce a~$C_4$ in~$G$. This contradicts the fact that~$G$ is chordal.

\clm{\label{clm:co-k13+2p1-chordal:s1-independent-set}
$S_1$ is an independent set.}
This holds as
no two vertices in~$S_1$ with a common non-neighbour in~$K$ are adjacent since~$K$ is maximum;
while no two vertices in~$S_1$ with different non-neighbours in~$K$ are adjacent since~$G$ is chordal.

\clm{\label{clm:co-k13+2p1-chordal:s-triangle-free}
$G[S]$ is a forest.}
Suppose for contradiction that~$G[S]$ is not a forest. Then, since~$G$ is
chordal,~$G[S]$ must contain a~$C_3$, on vertices $c_1,c_2,c_3$, say. By
Claim~\ref{clm:co-k13+2p1-chordal:s1-independent-set}, we may assume without
loss of generality that $c_2,c_3 \notin S_1$ and thus $c_2,c_3\in S_2$.
Then~$c_2$ and~$c_3$ must have the
same pair of non-neighbours $a,b \in K$ by
Claim~\ref{clm:co-k13+2p1-chordal:s2-same-k-non-nhbrs}. If $c_1 \in S_2$ then
by Claim~\ref{clm:co-k13+2p1-chordal:s2-same-k-non-nhbrs},
the non-neighbours of~$c_1$ in~$K$ are also~$a$ and~$b$.
If $c_1 \in S_1$ then by
Claim~\ref{clm:co-k13+2p1-chordal:s1-s2-common-non-nhbr},
the non-neighbour of~$c_1$ in~$K$ is either~$a$ or~$b$.
Hence, in both these cases, $K\setminus \{a,b\}\cup
\{c_1,c_2,c_3\}$ is a clique of size more that~$|K|$, contradicting the
maximality of~$K$.

\medskip
\noindent
We will consider two cases depending on whether or not~$G[S]$ is $2P_2$-free. We need two more claims to deal with these.
We are going to distinguish between two cases depending on whether or not~$G[S]$ is $2P_2$-free. For treating these two cases we need two more claims.

\clm{\label{clm:co-k13+2p1-chordal:K-comp-most-of-S}
If two vertices $s_1,s_2 \in S$, together with a vertex $w \in K$ form a triangle
then~$w$ is complete to $S \setminus (N(s_1) \cup N(s_2))$.}
Indeed, suppose for contradiction that $t \in S \setminus (N(s_1) \cup N(s_2))$ is not
adjacent to~$w$. Since $|K| \geq 7$, there must be vertices $x,y \in K$ that
are complete to $\{s_1,s_2,t\}$. Since~$t$ is non-adjacent to~$s_1$ and~$s_2$, we find that
$\{t,s_1,s_2,w,x,y\}$ induces a $\overline{K_{1,3}+2P_1}$, a contradiction.

\clm{\label{clm:co-k13+2p1-chordal:non-trivial-s-comp}
For any connected component~$D$ in~$G[S]$ that contains at least one edge,
there exist two vertices~$a$ and~$b$ in~$K$ such that $K\setminus \{a,b\}$
is complete to $S\setminus V(D)$.}
To see this, let~$D$ be a connected component with an edge~$st$. Since~$S_1$ is independent
by Claim~\ref{clm:co-k13+2p1-chordal:s1-independent-set},
we may assume
that $t\in S_2$. Let~$a$ and~$b$ in~$K$ be
the
two non-neighbours of~$t$.
It follows from Claims~\ref{clm:co-k13+2p1-chordal:s2-same-k-non-nhbrs}
and~\ref{clm:co-k13+2p1-chordal:s1-s2-common-non-nhbr} that~$a$ and~$b$ are the only possible
non-neighbours of~$s$ or~$t$ in~$K$. In other words,
$K\setminus \{a,b\}$
is complete to $\{s,t\}$, and hence to $S\setminus V(D)$ by Claim~\ref{clm:co-k13+2p1-chordal:K-comp-most-of-S}.

\medskip
\noindent
We are now ready to consider the two cases.

\thmcase{$G[S]$ contains an induced~$2P_2$.}
First suppose~$G[S]$ has only one connected component that contains an edge.
Then, since~$G[S]$ is a forest
and~$G[S]$ contains an induced~$2P_2$, deleting one vertex from~$S$, which we may do by Fact~\ref{fact:del-vert},
yields two connected components~$D_1$ and~$D_2$ that contain edges~$s_1t_1$ and~$s_2t_2$, respectively.
It follows from Claim~\ref{clm:co-k13+2p1-chordal:non-trivial-s-comp} that there exist vertices~$a$ and~$b$ in~$K$ such that
$S\setminus D_1$ is complete to $K \setminus \{a,b\}$.
In particular,~$s_2$ and~$t_2$ are complete to $K \setminus \{a,b\}$.
Hence, $K \setminus \{a,b\}$ is also complete to~$D_1$, by Claim~\ref{clm:co-k13+2p1-chordal:K-comp-most-of-S}.
Thus, $K \setminus \{a,b\}$ is complete to~$S$.
Deleting~$a$ and~$b$ (which we may do by Fact~\ref{fact:del-vert})
and applying a bipartite complementation between $K \setminus \{a,b\}$ and~$S$
(which we may do by Fact~\ref{fact:bip}) splits the graph into two disjoint parts:
a clique $G[K \setminus \{a,b\}]$, which has clique-width~$2$,
and a forest~$G[S]$, which has clique-width at most~$3$ by Lemma~\ref{lem:tree}.
We conclude that~$G$ has bounded clique-width.\footnote{We mean to say that the clique-width of~$G$
is bounded by a constant that does not depend on the size of~$G$ but only on the class of graphs under consideration. We allow this minor abuse of notation throughout the paper.}

\thmcase{$G[S]$ is $2P_2$-free.}
In this case~$S$ contains at most
one connected component with an edge. If such a connected component exists then it is a $2P_2$-free
tree, and hence it must be a~$P_2,K_{1,r}$ or a
double-star.
In all three cases
deleting at most two vertices from~$S$, which we may do by Fact~\ref{fact:del-vert},
yields a split graph.
If $S_2 \neq \emptyset$ then let~$s$ be a vertex in~$S_2$ and let~$k_1$ and~$k_2$ be its two (only) non-neighbours in~$K$.
By Claim~\ref{clm:co-k13+2p1-chordal:s2-common-non-nhbr}, any other vertex of~$S_2$ is non-adjacent to at least one of $k_1,k_2$.
Hence, after removing~$k_1$ and~$k_2$ (which we may do by Fact~\ref{fact:del-vert}), every vertex of~$S$ is adjacent to all but at most one vertex of~$K$.
(In the case where $S_2 = \emptyset$, we do not need to remove any vertices of~$K$.)
Next, we perform a bipartite complementation between~$K$ and~$S$, which we may do by Fact~\ref{fact:bip}.
This results in a new split graph in which each vertex of~$S$ is adjacent to at most one vertex of~$K$.
Hence, this graph, and consequently~$G$, has bounded clique-width by Fact~\ref{fact:2-conn}.
\qed
\end{proof}

\subsection{The Case $H=P_1+\overline{P_1+P_3}$}

We first prove three useful lemmas.

\begin{lemma}\label{lem:p1+paw-split}
The class of $(P_1+\overline{P_1+P_3})$-free split graphs has bounded clique-width.
\end{lemma}

\begin{proof}
Let~$G$ be an arbitrary $(P_1+\nobreak \overline{P_1+P_3})$-free split graph with split partition~$(C,I)$.
By Fact~\ref{fact:comp}, we may apply a subgraph complementation on the clique~$C$.
The resulting graph~$G'$ is bipartite. Because~$G$ is $(P_1+\nobreak \overline{P_1+P_3})$-free,~$G'$ is
$(P_2+\nobreak P_4)$-free and thus $S_{1,2,3}$-free.
Then the result follows from the fact that $S_{1,2,3}$-free bipartite graphs have bounded clique-width
by Lemma~\ref{lem:bipartite}.
\qed
\end{proof}

\begin{lemma}\label{lemma:p1+paw-chordal:paw-free}
Every connected $\overline{P_1+P_3}$-free chordal graph is a tree or a complete split graph.
\end{lemma}

\begin{proof}
Let~$G$ be a connected $\overline{P_1+P_3}$-free chordal graph.
By Lemma~\ref{lem:paw}, we find that~$G$ is $C_3$-free or complete multipartite. If~$G$ is $C_3$-free, then it must be a tree, since~$G$ is chordal. If~$G$ is complete multipartite, then at most
one partition class of~$G$ can contain more than one vertex, otherwise~$G$
would contain an induced~$C_4$. This means that~$G$ is a complete
split graph.
\qed
\end{proof}

Note that every induced $\overline{P_1+P_3}$ in a $(P_1+\nobreak \overline{P_1+P_3})$-free graph~$G$ is a dominating set of~$G$. The proof of the next lemma, in which disconnected graphs are considered, heavily relies on this fact.
We will also heavily exploit this property in the proof for the general case.

\begin{lemma}\label{lemma:p1+paw-chordal:disconnected}
The class of disconnected $(P_1+\nobreak \overline{P_1+P_3})$-free chordal graphs has bounded clique-width.
\end{lemma}

\begin{proof}
Let~$G$ be a disconnected $(P_1+\nobreak \overline{P_1+P_3})$-free chordal graph. Since~$G$ has at least two connected components and each connected
component therefore contains a~$P_1$, every connected component of~$G$ must be
$\overline{P_1+P_3}$-free. By Lemma~\ref{lemma:p1+paw-chordal:paw-free}, every
connected component of~$G$ must be a complete split graph or a tree. In the first case,
the clique-width of the connected component is readily seen to be at most~$2$. In the second case, the
clique-width of that connected component is at most~$3$ by Lemma~\ref{lem:tree}.
\qed
\end{proof}

We are now ready to prove our second result.

\begin{theorem}\label{thm:p1+paw-chordal}
The class of $(P_1+\nobreak \overline{P_1+P_3})$-free chordal graphs has bounded clique-width.
\end{theorem}

\begin{proof}
Let~$G$ be a $(P_1+\nobreak \overline{P_1+P_3})$-free chordal graph.
Let~$x$ be a
simplicial vertex in~$G$, which exists by Lemma~\ref{l-simplicial}.
Let $X=N(x)$ and $Y=V(G) \setminus(X \cup \{x\})$.
Note that no vertex of~$Y$ is adjacent to~$x$, so~$G[Y]$ must be
$\overline{P_1+P_3}$-free. By Lemma~\ref{lemma:p1+paw-chordal:paw-free}, every connected
component of~$G[Y]$ is either a tree or complete split graph.
We say that a connected component of~$G[Y]$ is {\em trivial} if it consists of a single
vertex. Otherwise it is {\em non-trivial}.

We will distinguish between two cases depending on whether or not~$G[Y]$ is $2P_2$-free. In the first case we will need the following claim.

\clm{\label{clm:p1+paw-chordal:non-triv-Y}
Suppose~$G[Y]$ contains at least two non-trivial components and $y \in Y$
is in such a component. If~$y$ is adjacent to $z \in X$ then~$y$ is complete
to~$X$ or~$z$ is complete to~$Y$.}
In order to prove this claim, suppose that~$y$ is not complete to~$X$. We will show that~$z$ is
complete to~$Y$. Let~$D$
be the connected component of~$G[Y]$ containing~$y$. Since~$y$
is not complete to~$X$, there must be a vertex $z' \in X$ that is not adjacent
to~$y$. Now $G[z,x,y,z']$ is a $\overline{P_1+P_3}$. Since~$G$ is
$(P_1+\nobreak \overline{P_1+P_3})$-free, we find that $\{x,y,z,z'\}$ must dominate~$G$.
No vertex of $Y \setminus V(D)$ is adjacent to~$x$ or~$y$. Therefore $Y
\setminus V(D)$ is dominated by $\{z,z'\}$.

Let~$y_1y_2$ be an edge in some non-trivial component~$D'$ of~$Y$ other than~$D$ (recall that such a component exists by our assumption).
If~$y_1$ and~$y_2$ are both adjacent to~$z'$ then $G[y,z',y_1,x,y_2]$ would be a $P_1+\nobreak \overline{P_1+P_3}$.
Therefore we may assume without loss of generality that~$y_1$ is not adjacent
to~$z'$. Since $\{z,z'\}$ dominates~$y_1$, we find that~$y_1$ must be adjacent to~$z$. If~$y_2$ is not adjacent to~$z$ then, since $\{z,z'\}$ dominates~$y_2$, we find that~$y_2$ must be adjacent to~$z'$. In this case $G[z,z',y_2,y_1]$ would be
a~$C_4$, contradicting the fact that~$G$ is chordal. Hence both~$y_1$ and~$y_2$ are adjacent to~$z$. Now $G[z,y_1,x,y_2]$ induces a
$\overline{P_1+P_3}$. Therefore~$z$ is complete to $Y\setminus D'$, since~$G$
is $(P_1+\nobreak \overline{P_1+P_3})$-free.
Recall that~$y_1$ is adjacent to~$z$ and
non-adjacent to~$z'$. By the same argument,
with~$y_1$ taking the role of~$y$,
since~$D$ is a non-trivial
component of~$G[Y]$, we find that~$z$ is complete to $Y \setminus V(D)$. Hence,~$z$ is
complete to~$Y$. This completes the proof of Claim~\ref{clm:p1+paw-chordal:non-triv-Y}.

\medskip
\noindent
We are now ready to consider the two possible cases.

\thmcase{\label{case:p1+paw-chordal:2p2}
$G[Y]$ contains an induced~$2P_2$.}
First suppose all vertices of this~$2P_2$ are in the same connected component~$D$ of~$G[Y]$. Since split graphs are $2P_2$-free by Lemma~\ref{lem:split}, we find that~$D$ is a tree by
Lemma~\ref{lemma:p1+paw-chordal:paw-free}. In this case, by
Fact~\ref{fact:del-vert}, we may delete one vertex in~$D$ so that the two edges
of the~$2P_2$ are in two different connected components of~$G[Y]$. We may therefore assume without loss of generality
that~$G[Y]$ contains two non-trivial components.

Let~$Y'$ be the set of vertices in~$Y$ that are in non-trivial components of~$G[Y]$. Let~$Y''$ be the set of vertices in~$Y'$ that are complete to~$X$. Let~$X'$ be the set of vertices in~$X$ that are complete to~$Y$. It follows from
Claim~\ref{clm:p1+paw-chordal:non-triv-Y} that $X\setminus X'$ is anti-complete
to $Y'\setminus Y''$. We can apply two bipartite complementation operations,
one between~$X'$ and $Y'\cup\{x\}$ and the other between $Y''\cup\{x\}$ and $X \setminus X'$.
This will separate $G[Y'\cup\{x\}]$ from the rest of the graph. By
Lemma~\ref{lemma:p1+paw-chordal:disconnected}, we find that $G[Y'\cup\{x\}]$ has bounded clique-width.
Because $G[V\setminus (Y' \cup \{x\})]$ is a $(P_1+\nobreak \overline{P_1+P_3})$-free split graph,
it has bounded clique-width by Lemma~\ref{lem:p1+paw-split}. By
Fact~\ref{fact:bip}, we find that~$G$ has bounded clique-width. This completes
the proof of Case~\ref{case:p1+paw-chordal:2p2}.

\thmcase{$G[Y]$ is $2P_2$-free.}
If~$G[Y]$ contains only trivial components then~$G$
is a $(P_1+\nobreak \overline{P_1+P_3})$-free split
graph, so it has bounded clique-width by
Lemma~\ref{lem:p1+paw-split}. Since~$G[Y]$ is $2P_2$-free, it can contain at
most one non-trivial component. We may therefore assume that~$G[Y]$ contains
exactly one non-trivial component~$D$.

First suppose that~$D$ is a tree. In this case~$G[D]$ must be a~$P_2$,
$K_{1,r}$ or a double-star.
In all three cases, deleting at most two vertices in~$D$ (which we may do by
Fact~\ref{fact:del-vert}) makes~$Y$ an independent set, in which case we argue
as before.
By Lemma~\ref{lemma:p1+paw-chordal:paw-free}, we may therefore assume that~$G[Y]$ is a complete split graph.
We can partition~$V(D)$ into two sets~$D_B$ and~$D_W$ such that~$D_B$ is a clique,~$D_W$
is an independent set and~$D_B$ is complete to~$D_W$ in~$G$.
We may assume $|D_B| \geq 3$. Indeed, if $|D_B| \leq 2$ then by
Fact~\ref{fact:del-vert} we may delete at most two vertices to obtain a graph
in which~$G[Y]$ has only trivial components, in which case we may argue as
before.

Let~$X'$ be the set of vertices in~$X$ that have neighbours in~$D$. We claim
that~$X'$ is complete to $Y \setminus V(D)$. Suppose for contradiction that $x'
\in X'$ is not adjacent to some vertex $y \in Y \setminus V(D)$. Then~$x'$ cannot
have two neighbours $y_1,y_2 \in D_B$ otherwise $G[y,x',y_1,x,y_2]$ would be a
$P_1+\nobreak \overline{P_1+P_3}$. Let $y_1 \in V(D)$ be a neighbour of~$x'$. Since $|D_B|
\geq 3$,~$x'$ must have two non-neighbours $y_2,y_3 \in D_B$. However, now
$G[y,y_1,y_2,x',y_3]$ is a $P_1+\nobreak \overline{P_1+P_3}$. This contradiction means
that~$X'$ is indeed complete to $Y \setminus V(D)$.

As~$X$ is a clique and~$X'$ is complete to $Y\setminus V(D)$,
we find that $(Y \setminus V(D)) \cup (X \setminus X')$ is complete to~$X'$. By
Fact~\ref{fact:bip}, we may apply a bipartite complementation between $(Y
\setminus V(D)) \cup (X \setminus X')$ and~$X'$ and another between $X \setminus
X'$ and~$\{x\}$. This separates $G[(Y \setminus V(D)) \cup (X \setminus X')]$ from
the rest of the graph, which is $G[\{x\} \cup X' \cup V(D)]$. The first graph is a
$(P_1+\nobreak \overline{P_1+P_3})$-free split graph, so it has bounded clique-width by
Lemma~\ref{lem:p1+paw-split}. It remains to show that $G[\{x\} \cup X' \cup V(D)]$
has bounded clique-width.

We partition the vertices of~$X'$ as follows: let~$Z$ be the set of vertices in~$X'$ that are complete to~$D_B$, let~$Z'$ be the set of vertices in $X'
\setminus Z$ that are complete to~$D_W$ and let $Z''=X'\setminus (Z \cup Z')$.
Let~$D_W'$ be the set of vertices in~$D_W$ that are complete to $Z' \cup Z''$
and let $D_W''=D_W \setminus D_W'$.

We claim that~$D_W''$ is anti-complete to~$Z''$. Suppose for contradiction that
$w \in D_W''$ is adjacent to $z \in Z''$. By definition,~$w$ must be
non-adjacent to some vertex $z' \in Z' \cup Z''$ and~$z$ must be non-adjacent to
some vertex $w' \in D_W$. Furthermore,~$z$ must be non-adjacent to some vertex
$b \in D_B$. Note that~$w$ is not adjacent to~$w'$ since~$D_W$ is independent.
Moreover,~$z$ and~$z'$ are adjacent because~$X'$ is a clique, and~$b$ is adjacent to both~$w$ and~$w'$ as~$D$ is a complete split graph.
Then~$b$ and~$z'$ must be non-adjacent, otherwise $G[b,w,z,z']$ would be a~$C_4$. Then~$w'$ must be adjacent to~$z'$, otherwise $G[w',z,x,w,z']$ would be a
$P_1+\nobreak \overline{P_1+P_3}$. However, this means that $G[z',z,w,b,w']$ induces a~$C_5$, contradicting the fact that~$G$ is chordal. Therefore~$D_W''$ is indeed
anti-complete to~$Z''$.

By Fact~\ref{fact:del-vert}, we may delete the vertex~$x$ from~$G$. Now
$D_B\cup Z'$ is complete to $D_W' \cup D_W'' \cup Z$, while~$Z''$ is complete to
$D_W' \cup Z$ and anti-complete to~$D_W''$. By Fact~\ref{fact:bip}, we may
apply two bipartite complementations: one between $Z' \cup D_B$ and $D_W' \cup
D_W'' \cup Z$ and the other between~$Z''$ and $D_W' \cup Z$. The resulting
graph will be partitioned into two disjoint graphs: $G[D_W \cup Z]$ and
$G[D_B \cup Z'\cup Z'']$. The first of these is a
$(P_1+\nobreak \overline{P_1+P_3})$-free split graph, so it has bounded clique-width by
Lemma~\ref{lem:p1+paw-split}.
Taking the complement of $G[D_B \cup Z'\cup Z'']$ (which we may do by
Fact~\ref{fact:comp}) yields the bipartite graph $\overline{G}[D_B \cup\nobreak  Z'\cup\nobreak
Z'']$, which is $2P_2$-free since~$G$ is chordal and therefore has bounded
clique-width by Lemma~\ref{lem:bipartite}. We conclude that~$G$ has bounded
clique-width.
This completes the proof of Theorem~\ref{thm:p1+paw-chordal}.\qed
\end{proof}

\subsection{The Case $H=P_1+\overline{2P_1+P_2}$}
A graph $G=(V,E)$ is {\em quasi-diamond-free} if its vertex set~$V$ can be partitioned into
a clique~$V_1$ and some other (possibly empty) set $V_2=V\setminus V_1$ so that~$G[V_2]$ is a $\overline{2P_1+P_2}$-free chordal graph,
every connected component of which has at most one neighbour in~$V_1$.

We prove the following lemma, which will play an important role in our proof.

\begin{lemma}\label{lem:p1+diamond-chordal:2-conn}
The class of quasi-diamond-free graphs has bounded clique-width.
\end{lemma}

\begin{proof}
Let~$G$ be a quasi-diamond-free graph with corresponding clique~$V_1$.
Let~$B$ be a block of~$G$. Then~$B$ is
either equal to~$V_1$ or contains at most one vertex of~$V_1$ with all its other
vertices belonging to~$V_2$. In the first case, the clique-width of~$B$ is at most~$2$. In the second case, we may delete the vertex of $B \cap V_1$ from~$B$ (if such a vertex exists)
by Fact~\ref{fact:del-vert}. This yields a $\overline{2P_1+P_2}$-free chordal graph~$G'$. By
Theorem~\ref{thm:co-k13+2p1-chordal}, we find that~$G'$ has bounded clique-width.
Therefore~$G$ has bounded clique-width by Fact~\ref{fact:2-conn}.\qed
\end{proof}
We are now ready to prove the following result.

\begin{theorem}\label{thm:p1+diamond-chordal}
The class of $(P_1+\nobreak \overline{2P_1+P_2})$-free chordal graphs has bounded clique-width.
\end{theorem}

\begin{proof}
Let $G=(V,E)$ be a $(P_1+\nobreak \overline{2P_1+P_2})$-free chordal graph.
We may assume without loss of generality that~$G$ is connected. Let~$v$ be a simplicial vertex in~$G$,
which exists by Lemma~\ref{l-simplicial}.
Let $L_1=N(v)$, $L_2=N(L_1)\setminus (L_1\cup \{v\})$ and $L_3=N(L_2)\setminus (L_2\cup L_1\cup \{v\})$.
Note that~$L_1$ is a clique, because~$v$ is simplicial.

\clm{\label{clm:p1+diamond-chordal:st-in-L23-nonadj}
If $s,t \in L_2 \cup L_3$ are non-adjacent then~$s$ is
adjacent to all but at most one vertex of~$N_{L_1}(t)$.}
Indeed, suppose for
contradiction that~$s$ is non-adjacent to distinct vertices $a,b \in N_{L_1}(t)$. Then
$G[s,a,b,t,v]$ is a $P_1+\nobreak \overline{2P_1+P_2}$, a contradiction.

\medskip
\noindent
Let~$x$ be a vertex of~$L_2$ such that $\Delta = |N_{L_1}(x)|$ is maximised.
Note that $G[V\setminus (\{v\}\cup L_1)]$ is $\overline{2P_1+P_2}$-free and that $\{v\} \cup L_1$ is a clique. Hence,
if $\Delta=1$ then we can apply Lemma~\ref{lem:p1+diamond-chordal:2-conn} to~$G$ with $V_1= \{v\} \cup L_1$. Thus,
from now on we may assume that $\Delta\ge 2$.
This means that~$x$ and~$v$ have at least two common neighbours in~$L_1$.
Hence, as~$G$ is $(P_1+\nobreak \overline{2P_1+P_2})$-free, we find that $$V=\{v\}\cup L_1 \cup L_2 \cup L_3.$$

\clm{\label{clm:p1+diamond-chordal:v=1}
Without loss of generality, every vertex in~$L_1$ has a neighbour in~$L_2$.}
In order to show this, let $L_1'\subseteq L_1$ be the set of vertices with no neighbour in~$L_2$.
We apply a bipartite complementation between $(L_1\setminus L_1')\cup \{v\}$ and~$L_1'$. We may do
so due to Fact~\ref{fact:bip}. As~$G[L_1']$ is a clique, it has clique-width at most~$2$, and we are left to consider
$G[V\setminus L_1']$.

\medskip
\noindent
As $ \Delta=|N_{L_1}(x)|$, we find that $\Delta\leq |L_1|$. We now consider two cases, depending on the difference between~$|L_1|$ and~$\Delta$.

\thmcase{\label{case:p1+diamond-chordal:delta-big} $\Delta\leq |L_1|-3$.}
Let~$z$ be an arbitrary vertex in $L_1\setminus N_{L_1}(x)$.
Let~$A_z$ be the set of neighbours of~$z$ in~$L_2$.
By Claim~\ref{clm:p1+diamond-chordal:v=1}, we find that $A_z\neq \emptyset$.

Let~$u$ be an arbitrary vertex in~$A_z$.
By our choice of~$x$, we have that $|N_{L_1}(u)| \leq |N_{L_1}(x)|$, and so~$u$ must have a non-neighbour $y_u \in N_{L_1}(x)$.
Then~$u$ is non-adjacent to~$x$ otherwise $G[u,x,y_u,z]$ would be a~$C_4$,
contradicting the fact that~$G$ is chordal. Now by Claim~\ref{clm:p1+diamond-chordal:st-in-L23-nonadj},
we find that
$$N_{L_1}(u)=(N_{L_1}(x)\setminus \{y_u\}) \cup \{z\}.$$
The above implies that $A_z\cap A_{z'}=\emptyset$ for all $z,z'\in L_1\setminus N_{L_1}(x)$.

We now show that $y_u=y_{u'}$ for any two vertices $u\in A_z$ and $u'\in A_{z'}$ and for any two (not necessarily distinct) vertices $z,z'\in L_1\setminus N_{L_1}(x)$.
First, suppose $z,z' \in L_1\setminus N_{L_1}(x)$ are distinct. Such vertices exist since $\Delta\leq |L_1|-3$. Let $u\in A_z$ and $u'\in A_{z'}$. We may assume such vertices exist since~$A_z$ and~$A_{z'}$ are not empty because of Claim~\ref{clm:p1+diamond-chordal:v=1}.
If $y_u \neq y_{u'}$ then, since 
$y_u,y_{u'} \in N_{L_1}(x)$ and $z,z' \in L_1 \setminus N_{L_1}(x)$, we find
that $y_u, y_{u'},z$ and $z'$ are distinct vertices in~$L_1$.
Since $N_{L_1}(u)=(N_{L_1}(x)\setminus \{y_u\}) \cup
\{z\}$ and $N_{L_1}(u')=(N_{L_1}(x)\setminus \{y_{u'}\}) \cup \{z'\}$, we find that~$u$ is adjacent
to~$y_{u'}$ and~$z$, but~$u'$ is non-adjacent to both~$y_{u'}$ and~$z$.
Therefore Claim~\ref{clm:p1+diamond-chordal:st-in-L23-nonadj} implies that~$u$ and~$u'$ must be adjacent;
however then $G[u,u',y_u,y_{u'}]$ is a~$C_4$, a contradiction.
Hence, $y_u=y_{u'}$.
Since the $u$-vertices in different sets~$A_z$ and~$A_{z'}$ share the same~$y$ vertex, and there are at least two such sets, this immediately implies that $u$-vertices from the same set~$A_z$ also share the same $y$-vertex.
Thus there exists a vertex $y^*\in N_{L_1}(x)$
such that for every $u\in A_z$ and every $z\in L_1\setminus N_{L_1}(x)$, we have
$$N_{L_1}(u)=(N_{L_1}(x)\setminus \{y^*\}) \cup \{z\}.$$
Let $A=N_{L_1}(x)\setminus \{y^*\}$. Let~$A_{y^*}$ be the set of vertices in~$L_2$ whose neighbourhood in~$L_1$ is~$N_{L_1}(x)$ (so $x\in A_{y^*})$.
Now for each vertex $z \in L_1\setminus A$ (including the case where $z=y^*$) and every $u\in A_z$, we have
$$N_{L_1}(u)=A \cup \{z\}.$$
Let~$X$ be the set of vertices~$u \in L_2 \cup L_3$
whose neighbourhood in~$L_1$ is properly contained in~$N_{L_1}(x)$, that is, for which
$N_{L_1}(u) \subsetneq N_{L_1}(x)=A\cup \{y^*\}$.
Note that, as no vertex in~$L_3$ has a neighbour in~$L_1$, we have $L_3\subseteq X$. Also note
that the sets~$X$ and~$A_z$, $z\in L_1\setminus A$ form a partition of $L_2\cup L_3$.

Consider two distinct vertices $w_1,w_2 \in L_1 \setminus A$.
Note that~$w_1$ and~$w_2$ are not necessarily distinct from~$y^*$, but
at least one of $w_1,w_2$ is distinct from~$y^*$.
Also note that if a vertex $u \in X$ is adjacent to~$w_i$ $(i=1,2)$
then $w_i=y^*$.

Suppose there is a path~$P$ in $G[L_2 \cup L_3]$ from some
vertex $t_1 \in A_{w_1}$ to some vertex $t_2 \in A_{w_2}$. We shall choose~$P$ such that~$|V(P)|$ is minimum, where the minimum is taken over all choices of
$w_1,w_2,t_1,t_2$ and~$P$. It follows from the minimality of~$P$ that
$V(P)\setminus \{t_1,t_2\}\subseteq X$.
Moreover, since $N_{L_1}(t_1)=A \cup \{w_1\}$ and $N_{L_1}(t_2)=A \cup \{w_2\}$, it follows that~$w_1$ and~$w_2$ are non-adjacent to~$t_2$ and~$t_1$, respectively.
Thus,~$t_1$ and~$t_2$ must be non-adjacent, as otherwise $G[t_1,t_2,w_2,w_1]$ would be a~$C_4$.

Without loss of generality, we may assume $w_1 \neq y^*$.
Since $V(P)\setminus \{t_1,t_2\}\subseteq X$, we find that~$w_1$ must be anti-complete to $V(P) \setminus \{t_1,t_2\}$.
Let~$t_3$ be the neighbour of~$w_2$ on~$V(P)$ that is nearest to~$t_1$.
(If~$w_2$ has no neighbours in $V(P) \setminus \{t_1,t_2\}$ then $t_3=t_2$.)
Note that $t_3 \neq t_1$, since~$w_2$ is not adjacent to~$t_1$.
Let~$P'$ be the part of the path~$P$ from~$t_1$ to~$t_3$.
The only neighbour of~$w_1$ in~$V(P')$ is~$t_1$.
The only neighbour of~$w_2$ in~$V(P')$ is~$t_3$.
Since~$w_1$ and~$w_2$ are adjacent and~$P'$ is an induced path on at least two vertices it follows that $G[V(P') \cup \{w_1,w_2\}]$ is a cycle on at least four vertices, contradicting the fact that~$G$ is chordal.
We have so far shown that
for any two distinct vertices $w_1,w_2 \in L_1 \setminus A$, there is no path in
$G[L_2 \cup L_3]$ from any vertex of~$A_{w_1}$ to any vertex of~$A_{w_2}$.

Now suppose that $u\in X$. As $\Delta\leq |L_1|-3$, there exist three
pairwise distinct vertices $w_1,w_2,w_3 \in L_1 \setminus A$.
By Claim~\ref{clm:p1+diamond-chordal:v=1}, there exist three vertices $t_1 \in
A_{w_1}$, $t_2 \in A_{w_2}$ and $t_3 \in A_{w_3}$. Because the sets~$A_{w_i}$ are mutually disjoint,
$t_1,t_2$ and~$t_3$ are also distinct. It follows from the conclusion above
that~$u$ can be adjacent to at most one of $t_1,t_2$ and~$t_3$.
Without loss of generality, assume that~$u$ is non-adjacent to~$t_1$ and~$t_2$
and that $w_1\neq y^*$. Then~$u$ cannot be adjacent to~$w_1$.
By Claim~\ref{clm:p1+diamond-chordal:st-in-L23-nonadj},~$u$ must be adjacent to every vertex of~$A$.
Since $N_{L_1}(u) \subsetneq N_{L_1}(x)$, it follows that $N_{L_1}(u)=A$.
Since~$u$ was an arbitrary vertex in~$X$, together with the observations made earlier, this shows that every vertex in $L_2\cup L_3$
is adjacent to every vertex of~$A$ and at most one other vertex in~$L_1$.
Since $\Delta \geq 2$, we have that $|A| \geq 1$, and so~$L_3$ must be empty.
By Fact~\ref{fact:bip}, we may apply a bipartite complementation between~$A$ and~$L_2$
after which we may apply Lemma~\ref{lem:p1+diamond-chordal:2-conn}.
This completes the proof of Case~\ref{case:p1+diamond-chordal:delta-big}.

\thmcase{\label{case:p1+diamond-chordal:delta-small} $\Delta\geq |L_1|-2$.}
Since $|L_1| \geq |N_{L_1}(x)|=\Delta$, there are at most two vertices in
$L_1\setminus N_{L_1}(x)$. By Fact~\ref{fact:del-vert}, we may delete these
vertices, if they exist. Note that this changes neither the value of~$\Delta$
nor the choice of~$x$. Therefore we may assume that $L_1 = N_{L_1}(x)$.

Then $N_{L_1}(w) \subseteq N_{L_1}(x)$ for all $w \in L_2$.
If $\Delta=|L_1| \leq 3$
then by deleting at most two vertices of~$L_1$ (which we may do by Fact~\ref{fact:del-vert})
we obtain a new graph for which we may apply Lemma~\ref{lem:p1+diamond-chordal:2-conn}.
We may therefore assume without loss of generality that $\Delta \geq 4$.

We distinguish three subcases depending on whether or not~$x$ dominates~$L_2$
and whether or not~$G[L_2]$ is a clique.

\thmsubcase{\label{subcase:p1+diamond-chordal:x-not-dom-L_2} $x$ does not dominate~$L_2$.}
Let $y \in L_2$ be a non-neighbour of~$x$. Recall
that $N_{L_1}(x)=L_1$. By Claim~\ref{clm:p1+diamond-chordal:st-in-L23-nonadj} we find that~$y$ must be adjacent to all but at most one vertex of~$L_1$. If~$y$ is not adjacent
to some vertex of~$L_1$, we may delete this vertex by Fact~\ref{fact:del-vert}.
We may therefore assume that $\Delta \geq 3$ and that~$y$ is complete to~$L_1$.

Suppose $w \in L_2$ has two non-neighbours $a,b \in N_{L_1}(x)$.
As $\{x,y\}$ is complete to~$L_1$, it follows that~$w$ is adjacent to
both~$x$ and~$y$ by Claim~\ref{clm:p1+diamond-chordal:st-in-L23-nonadj}. However, then $G[x,w,y,a]$ is a~$C_4$, contradicting
the fact that~$G$ is chordal. Therefore, every vertex in~$L_2$ is adjacent to
all but at most one vertex of~$L_1$. In particular,
as $\Delta\geq 3$,
every vertex in~$L_2$ has at least two neighbours in~$L_1$. This fact,
together with the fact that no vertex in~$L_3$ has neighbours in~$L_1$ and
Claim~\ref{clm:p1+diamond-chordal:st-in-L23-nonadj}, implies that
every vertex of~$L_2$ is adjacent to every vertex of~$L_3$.
By applying a bipartite complementation between~$L_2$ and~$L_3$,
we separate~$G[L_3]$ from
$G[V\setminus L_3]=G[\{v\}\cup L_1 \cup L_2]$.
Note that~$G[L_3]$ is a $\overline{2P_1+P_2}$-free chordal graph, so it has bounded
clique-width by Theorem~\ref{thm:co-k13+2p1-chordal}.
By Fact~\ref{fact:bip}, we may therefore assume that $L_3=\emptyset$.

Let~$X$ be the set of vertices in~$L_2$ that are complete to~$L_1$.
For $z \in L_1$, let~$U_z$ be the set of vertices in~$L_2$ that
are complete to $L_1 \setminus \{z\}$ and non-adjacent to~$z$.
As every vertex in~$L_2$ is adjacent to all but at most one vertex of~$L_1$, we find that the sets~$X$ and~$U_z$, $z\in L_1$, form a partition of vertices of~$L_2$.

Suppose there are at most six vertices $z \in L_1$ such that~$U_z$ is not empty.
By Facts~\ref{fact:del-vert} and~\ref{fact:bip}, we may apply a
bipartite complementation between~$L_1$ and~$L_2$ and then delete these vertices.
In the resulting graph, no vertex of~$L_2$ has a neighbour in~$L_1$ and we
can apply Lemma~\ref{lem:p1+diamond-chordal:2-conn}. We may therefore assume
that there are at least seven vertices $z \in L_1$ such that~$U_z$ is not empty.

Consider two distinct vertices $z_1,z_2 \in L_1$. We claim that~$U_{z_1}$ must be
anti-complete to~$U_{z_2}$. Indeed, if $y_1 \in U_{z_1}$ were adjacent to $y_2
\in U_{z_2}$ then $G[y_1,y_2,z_1,z_2]$ would be a~$C_4$, contradicting the fact
that~$G$ is chordal.

We will now show that by deleting at most one vertex from~$L_2$ (which we may
do by Fact~\ref{fact:del-vert}), we can make~$G[L_2]$ into a $P_3$-free graph.
Indeed, suppose that~$G[L_2]$ contains an induced~$P_3$ on vertices $v_1,v_2,v_3$.

First, consider a vertex $z \in L_1$ such that $v_1,v_2,v_3 \notin U_z$ and $U_z$ is non-empty. Suppose
$y \in U_z$. Then~$y$ must have at least one neighbour in $\{v_1,v_2,v_3\}$,
otherwise $G[y,v_2,z,v_1,v_3]$ would be a $P_1+\nobreak \overline{2P_1+P_2}$. Since
there are at least seven non-empty sets~$U_z$, there must be at least four non-empty sets~$U_z$
that do not contain a vertex in $\{v_1,v_2,v_3\}$.
Therefore there must be two sets~$U_{z_1}$ and~$U_{z_2}$ containing vertices~$y_1$ and~$y_2$, respectively, such that~$y_1$ and~$y_2$ are adjacent to the same vertex
in $\{v_1,v_2,v_3\}$, say~$v_i$.
Since~$U_{z_1}$ and~$U_{z_2}$ are anti-complete,~$y_1$ and~$y_2$ are non-adjacent. Hence,
$G[y_1,v_i,y_2]$ is a~$P_3$.
Also note that $v_i \in X$ since~$v_i$ has neighbours in both~$U_{z_1}$ and~$U_{z_2}$.

Now
let $z_3 \in L_1 \setminus \{z_1,z_2\}$
and suppose $y_3 \in U_{z_3}$. By the same argument as above,~$y_3$ must have a
neighbour in $\{y_1,v_i,y_2\}$. Moreover,
as~$U_{z_3}$ is anti-complete to both~$U_{z_1}$ and~$U_{z_2}$,
we find that~$y_3$ is non-adjacent to both~$y_1$ and~$y_2$. Hence,~$y_3$ must be adjacent to~$v_i$.
Now choose $z_4,z_5 \in L_1 \setminus \{z_1,z_2\}$ with $y_4 \in U_{z_4}$ and $y_5 \in
U_{z_5}$. Such vertices exist by our earlier assumption. By the same argument,
$G[y_4,v_i,y_5]$ is a~$P_3$, so~$v_i$ is complete to~$U_{z_3}$ for every
$z_3 \in L_1 \setminus \{z_4,z_5\}$. Hence,~$v_i$ is complete to~$U_z$ for every
$z \in L_1$. This implies that, if~$G[L_2]$ contains a~$P_3$, then some vertex of this~$P_3$
is adjacent to every vertex of every set~$U_z$.

Suppose that there exist two vertices $v',v'' \in L_2$ that are both complete to every vertex of every set~$U_z$.
Choose $y_1 \in U_{z_1}$ and $y_2 \in U_{z_2}$ with~$z_1$ and~$z_2$ distinct.
Note that~$y_1$ and~$y_2$ are non-adjacent and so $y_i\notin \{v',v''\}$ for $i=1,2$.
So, $\{v',v''\}$ is complete to $\{y_1,y_2\}$ by the assumption on~$v'$ and~$v''$.
If~$v'$ and~$v''$ are non-adjacent, then $G[v',y_1,v'',y_2]$ is a~$C_4$; if~$v'$ and~$v''$ are adjacent, then $G[v,v',v'',y_1,y_2]$
is a $P_1+\nobreak \overline{2P_1+P_2}$. In either case we have a contradiction, since~$G$
is a $(P_1+\nobreak \overline{2P_1+P_2})$-free chordal graph.
We have thus showed that there exists
at most one vertex that is complete to all~$U_z$.
This implies that if~$G[L_2]$ contains an induced~$P_3$ then there is a unique vertex in~$L_2$ that is present in every induced~$P_3$ in~$G[L_2]$.

By Fact~\ref{fact:del-vert} we may delete the vertex that is on every induced~$P_3$ (if~$G$ is not~$P_3$-free already). In this way we change~$G[L_2]$ into a $P_3$-free graph, which means that each connected component of~$G[L_2]$ is now a complete graph.

Consider an arbitrary connected component~$K$ of~$G[L_2]$.
As every vertex in~$L_2$, and thus in~$V(K)$, is adjacent to all but at most one vertex in~$L_1$ and
as~$G$ is $C_4$-free, we find that either~$V(K)$ is complete to~$L_1$ or to $L_1\setminus z$ for some
$z\in L_1$. Let~$G'$ be the graph obtained from~$G$ by performing a bipartite
complementation between~$L_1$ and~$L_2$.
Then every component in~$G[L_2]$ has, in~$G'$, at most one neighbour in~$L_1$.
Case~\ref{subcase:p1+diamond-chordal:x-not-dom-L_2} now follows directly from Fact~\ref{fact:bip}
and Lemma~\ref{lem:p1+diamond-chordal:2-conn}.

\thmsubcase{\label{subcase:p1+diamond-chordal:L_2-is-a-clique} $G[L_2]$ is a clique.}
In this case we may assume that there is a vertex $x' \in L_2 \setminus \{x\}$ that
has at least two neighbours in~$L_1$, as otherwise we could delete~$x$
(which we may do by Fact~\ref{fact:del-vert}) and apply Lemma~\ref{lem:p1+diamond-chordal:2-conn}.
Recall that, by definition,~$L_3$ has no neighbours in~$L_1$. Because both~$x$ and~$x'$ have at least
two neighbours in~$L_1$, Claim~\ref{clm:p1+diamond-chordal:st-in-L23-nonadj}
tells us that $\{x,x'\}$ is complete to~$L_3$.

If $y \in L_2$ is non-adjacent to $z \in L_3$ then $G[v,x,x',y,z]$ is a $P_1+\nobreak \overline{2P_1+P_2}$,
since~$G[L_2]$ is a clique. So,~$L_2$ is complete to~$L_3$.
By Fact~\ref{fact:bip} we may apply a bipartite complementation between~$L_2$ and~$L_3$, after which~$G[L_3]$ will be disconnected from the rest of the graph (since $V=\{v\}\cup L_1\cup L_2 \cup L_3$ and~$L_3$ is anti-complete to $\{v\}\cup L_1$).
Since~$G[L_3]$ is a $\overline{2P_1+P_2}$-free chordal graph, it has bounded clique-width.
So, it remains to show that $G[V\setminus L_3]=G[\{v\}\cup L_1\cup L_2]$ has bounded clique-width.
Now $G[\{v\} \cup L_1]$ and~$G[L_2]$ are cliques. Moreover, as~$G$ is chordal,~$G$ is $C_4$-free.
Applying a complementation to the whole graph (which we may do by Fact~\ref{fact:comp}) gives a $2P_2$-free bipartite graph, which has bounded clique-width by Lemma~\ref{lem:bipartite}.

\thmsubcase{ $x$ dominates~$L_2$, but~$G[L_2]$ is not a clique.}
Since~$G[L_2]$ is $\overline{2P_1+P_2}$-free, $G[L_2 \setminus \{x\}]$ must be
$P_3$-free. In other words, each connected component of $G[L_2\setminus
\{x\}]$ is a clique.

Since~$G[L_2]$ is not a clique, $G[L_2 \setminus \{x\}]$ must contain at least
two cliques, so deleting~$x$ from~$G$ (which we may do by
Fact~\ref{fact:del-vert}) means that~$G[L_2]$ no longer has a dominating
vertex. Note that this deletion may change the value of~$\Delta$. By the
same arguments as at the start of the proof, we may assume that $\Delta \geq 2$
and so $V(G)=\{v\} \cup L_1 \cup L_2 \cup L_3$. Again, by
Claim~\ref{clm:p1+diamond-chordal:v=1}, we may assume that every vertex of~$L_1$ has a neighbour in~$L_2$ in~$G$. Then if $\Delta\leq |L_1|-3$, we may
apply Case~\ref{case:p1+diamond-chordal:delta-big}. We may therefore assume
that $\Delta\geq |L_1|-2$. By the same arguments as at the start of
Case~\ref{case:p1+diamond-chordal:delta-small}, we may assume that
$|L_1|=\Delta$ and $\Delta \geq 4$. To make this assumption, we may have to
delete vertices from~$L_1$, which could cause vertices that were in~$L_2$
previously to now be in~$L_3$ for this modified graph. However, at no point
above do we add vertices to~$L_2$, so it is still the case that every component
of~$G[L_2]$ is a clique. Therefore
Case~\ref{subcase:p1+diamond-chordal:L_2-is-a-clique} or
Case~\ref{subcase:p1+diamond-chordal:x-not-dom-L_2} applies, depending on
whether~$G[L_2]$ now contains one or more components, respectively. This
completes the proof of Theorem~\ref{thm:p1+diamond-chordal}.\qed
\end{proof}

\subsection{The Case $H=\overline{S_{1,1,2}}$}\label{sec:co-chair}

We now show that the clique-width of $\overline{S_{1,1,2}}$-free chordal graphs
is bounded. Switching to the complement, we study $S_{1,1,2}$-free co-chordal
graphs, which are a subclass of $(2P_2,C_5,S_{1,1,2})$-free graphs.
First, in Lemma~\ref{lem:thinspider}, we show that prime $(2P_2,C_5,S_{1,1,2})$-free graphs are thin spiders if they contain an induced net.
We then use this lemma in combination with the two prime extension lemmas from Section~\ref{sec:prelim} (Lemmas~\ref{lem:fig:diamond-prime-ext} and~\ref{lem:cogemprimeext})
to provide, in Lemma~\ref{lem:primecochairfrchordal}, a structural description of prime $\overline{S_{1,1,2}}$-free chordal graphs.
Finally, in Theorem~\ref{thm:co-chair}, we use this structural description
to show boundedness of the clique-width of $\overline{S_{1,1,2}}$-free chordal graphs.

\begin{lemma}\label{lem:thinspider}
If a prime $(2P_2,C_5,S_{1,1,2})$-free graph~$G$ contains an induced subgraph
isomorphic to the net (see \figurename~\ref{fig:co-gem-prime-ext}) then~$G$ is a thin spider.
\end{lemma}

\begin{proof}
Suppose that~$G$ is a prime $(2P_2,C_5,S_{1,1,2})$-free graph and suppose
that~$G$ contains a net, say~$N$ with vertices $a_1,a_2,a_3,b_1,b_2,b_3$ such
that $a_1,a_2,a_3$ is an independent set (the {\em end-vertices} of~$N$),
$b_1,b_2,b_3$ is a clique (the {\em mid-vertices} of~$N$), and the only edges
between $a_1,a_2,a_3$ and $b_1,b_2,b_3$ are $a_ib_i \in E(G)$ for $i \in\{1,2,3\}$.

Let $M=V(G) \setminus V(N)$. We partition~$M$ as follows: For $i \in
\{1,\dots,5\}$, let~$M_i$ be the set of vertices in~$M$ with exactly~$i$
neighbours in~$V(N)$.
Let~$U$ be the set of vertices in~$M$ adjacent to every vertex
of~$V(N)$. 
Let~$Z$ be the set of vertices in~$M$ with no neighbours
in~$V(N)$. 
Note that~$Z$ is an independent set in~$G$, since~$G$ is
$2P_2$-free.

We now analyse the structure of~$G$ through a series of claims.

\clm{\label{no125v} $M_1 \cup M_2 \cup M_5 = \emptyset$.}
First suppose $x \in M_1 \cup M_2$. By symmetry, we may assume that~$x$ is
adjacent to at least one vertex in $\{a_1,b_1\}$ and anti-complete to
$\{a_2,b_2\}$. If~$x$ is adjacent to~$a_1$ then $G[x,a_1,a_2,b_2]$ is a~$2P_2$.
Therefore~$x$ is adjacent to~$b_1$, but not to~$a_1$. However this means that
$G[b_1,a_1,x,b_2,a_2]$ is an~$S_{1,1,2}$.
We conclude that $M_1 \cup M_2=\emptyset$.

Now suppose $x \in M_5$. We may assume by symmetry that~$x$ is non-adjacent to~$a_1$ or~$b_1$. Then $G[x,a_2,a_3,b_1,a_1]$ is an~$S_{1,1,2}$. It follows that $M_5=\emptyset$, completing the proof of Claim~\ref{no125v}.

\medskip
\noindent
Next, we prove that the vertices in~$M_3$ and~$M_4$ have a restricted type of neighbourhood in~$V(N)$:

\clm{\label{restricted3v}
Every $x \in M_3$ is adjacent to either exactly one end-vertex~$a_i$ and its two opposite mid-vertices~$b_j$ and~$b_k$ ($j \neq i$, $k \neq i$) or to all three mid-vertices of~$N$.}
Suppose that $x\in M_3$ is non-adjacent to at least one mid-vertex.
If~$x$ is adjacent to at least two end-vertices, say~$a_1$ and $a_2$, then~$x$ must be adjacent to~$b_1$ or~$b_2$, otherwise $G[x,a_1,b_1,b_2,a_2]$ would be a~$C_5$.
By symmetry we may assume that~$x$ is adjacent to~$b_1$.
As $x\in M_3$, this means that $G[x,a_1,b_2,b_3]$ is a~$2P_2$.
Hence, by symmetry,~$x$ must be adjacent to exactly one end-vertex, say~$a_1$, and two mid-vertices.
If~$x$ is non-adjacent to~$b_2$ then $G[a_1,x,a_2,b_2]$ is a~$2P_2$.
By symmetry,~$x$ must therefore be adjacent to $b_2$ and~$b_3$, completing the proof of Claim~\ref{restricted3v}.

\medskip
\noindent
The situation for~$M_4$ is similar to that of~$M_3$, as shown in the following claim.

\clm{\label{restricted4v} If $x \in M_4$ then it is adjacent to exactly one end-vertex and all mid-vertices.}
Let~$x \in M_4$. Without loss of generality,~$x$ must be adjacent to an end-vertex, say~$a_1$.
If~$x$ is adjacent to all three end-vertices $a_1,a_2,a_3$ and, say,~$b_1$ then $G[x,a_2,b_2,b_3,a_3]$ is a~$C_5$.
If~$x$ is adjacent to exactly two end-vertices, say~$a_1$ and~$a_2$, then $G[x,a_1,a_2,b_3,a_3]$ is an~$S_{1,1,2}$ unless $x$ is non-adjacent to $b_3$.
However, if~$x$ is non-adjacent to~$b_3$ then $G[a_1,x,b_3,a_3]$ is a~$2P_2$. 
Hence, $x$ must be adjacent to exactly one end-vertex. Consequently, as $x\in M_4$, we find that $x$ is adjacent to all three mid-vertices of~$N$.
This completes the proof of Claim~\ref{restricted4v}.

\medskip
\noindent
Let~$Mid_3$ denote the set of vertices in~$M_3$ that are adjacent to all three mid-vertices of~$N$ (and non-adjacent to any end-vertex of $N$).

\clm{\label{UM34join} $U$ is complete to $(M_3 \cup M_4)$.}
Suppose that $u \in U$ and $x \in (M_3 \cup M_4)$ are not adjacent.
If $x \in Mid_3$ then $G[u,a_2,a_3,b_1,x]$ is an~$S_{1,1,2}$.
If $x \in (M_3 \cup M_4) \setminus Mid_3$, then without loss of generality~$x$
is adjacent to~$a_1$ and $G[u,a_2,a_3,a_1,x]$ is an~$S_{1,1,2}$. This completes the proof of Claim~\ref{UM34join}.

\medskip
\noindent
Let~$Z_1$ denote the set of vertices in~$Z$ that have a neighbour in $M_3 \cup M_4$, and let $Z_0=Z \setminus Z_1$.

\clm{\label{Z1cojoin} $Z_1$ is anti-complete to $((M_3 \cup M_4) \setminus Mid_3)$.}
Suppose that $z \in Z_1$ and $x \in (M_3 \cup M_4) \setminus Mid_3$ are adjacent.
Without loss of generality, we may assume that~$x$ is adjacent to~$a_1$ and~$b_3$.
Then $G[x,a_1,z,b_3,a_3]$ is an~$S_{1,1,2}$. This completes the proof of Claim~\ref{Z1cojoin}.

\medskip
\noindent
Thus, the only possible neighbours of~$Z_1$ vertices in $M_3 \cup M_4$ are the vertices in~$Mid_3$.

\clm{\label{UZ1join} $U$ is complete to $Z_1$.}
Suppose $u \in U$ and $z \in Z_1$ are non-adjacent. 
By the definition of~$Z_1$, the vertex~$z$ has a neighbour $x \in M_3\cup M_4$.
By Claim~\ref{Z1cojoin}, it follows that $x\in Mid_3$.
By Claim~\ref{UM34join},~$x$ must be adjacent to~$u$.
Then $G[u,a_2,a_3,x,z]$ is an~$S_{1,1,2}$. This completes the proof of Claim~\ref{UZ1join}.

\medskip
\noindent
Let $X=V(N) \cup M_3 \cup M_4 \cup Z_1$.
Then~$X$ is a module: every vertex in~$U$ is complete to~$X$ (due to the definition of~$U$, together with Claims~\ref{UM34join} and \ref{UZ1join})
and every vertex in~$Z_0$ is anti-complete to~$X$ (due to the definitions of $Z,Z_0$ and~$Z_1$, together with the fact that~$Z$ is an independent set).
Since~$G$ is prime,~$X$ must be a trivial module.
Since~$X$ contains more than one vertex, it follows that $V(G)=X=V(N) \cup M_3 \cup M_4 \cup Z_1$.
Hence $U \cup Z_0=\emptyset$.

It remains to show that $G=G[V(N) \cup M_3 \cup M_4 \cup Z_1]$ is a thin spider.
For $i \in \{1,2,3\}$ let $M'_i=(M_3 \cup M_4) \cap N(a_i)$. Note that $M_3\cup M_4=Mid_3\cup M_1'\cup M_2'\cup M_3'$.
The next two claims show how each~$M'_i$ is connected to other subsets of~$V(G)$.

\clm{\label{Mijjoin} For $i \neq j$, $M'_i$ is complete to $M'_j$.}
By symmetry we may assume that $i=1$ and $j=2$.
If $x \in M'_1$ is non-adjacent to $y \in M'_2$ then, by Claims~\ref{restricted3v} and~\ref{restricted4v}, we find that $G[x,a_1,y,a_2]$ is a~$2P_2$.
This completes the proof of Claim~\ref{Mijjoin}.

\clm{\label{MijoinMid3} For every $i=1,2,3$, $M'_i$ is complete to $Mid_3$.}
By symmetry we may assume that $i=1$.
If $x \in M'_1$ is non-adjacent to $y \in Mid_3$ then, by
Claims~\ref{restricted3v} and~\ref{restricted4v}, we find that
$G[b_2,a_2,y,x,a_1]$ is an~$S_{1,1,2}$.
This completes the proof of Claim~\ref{MijoinMid3}.

\medskip
\noindent
By Claims~\ref{restricted3v}, \ref{restricted4v}, \ref{Z1cojoin},~\ref{Mijjoin}
and~\ref{MijoinMid3} we find that, for every $i\in\{1,2,3\}$, $M'_i \cup
\{b_i\}$ is a module, so $M'_i = \emptyset$ (since~$G$ is prime).
Consequently, $V(G)=V(N) \cup Mid_3 \cup Z_1$.
Next, we show the following:

\clm{\label{Mid3clique} $Mid_3$ is a clique.}
Suppose that~$Mid_3$ is not a clique.
Let~$Q$ be the vertex set of a component of~$\overline{G[Mid_3]}$,
such that~$\overline{G[Q]}$, contains an edge (so~$G[Q]$ contains a non-edge).
Since~$G$ is prime,~$Q$ cannot be a module in~$G$.
Note that, in~$G$, the set $Mid_3 \setminus Q$ is complete to~$Q$.
Moreover, every vertex in $Q\subseteq Mid_3$ is adjacent to every mid-vertex of~$N$ and non-adjacent to every end-vertex of~$N$ (by definition).
Hence there must be vertices $x,y \in Q$ and $z \in Z_1$ such that~$z$ distinguishes~$x$ and~$y$, say~$z$ is adjacent to~$x$ in~$G$, but not to~$y$.
Because~$\overline{G[Q]}$ is connected, we may assume that~$x$ and~$y$ are adjacent in~$\overline{G}$, in which case~$x$ and~$y$ are non-adjacent in~$G$.
However, then $G[b_3,a_3,y,x,z]$ is an~$S_{1,1,2}$. This completes the proof of Claim~\ref{Mid3clique}.

\medskip
\noindent
By Claim~\ref{Mid3clique} and the definition of~$Mid_3$, we find that $\{b_1,b_2,b_3\} \cup Mid_3$ is a clique.
By the definition of $Z$ and the fact that $Z$ is independent, $\{a_1,a_2,a_3\} \cup Z_1$ is
an independent set. Therefore~$G$ is a split graph. By
Lemma~\ref{lem:chair-split-spider}, since~$G$ is prime and $S_{1,1,2}$-free, it
must be a spider. Since~$G$ contains an induced net, it must be a thin
spider.
\qed
\end{proof}

\begin{lemma}\label{lem:primecochairfrchordal}
If~$G$ is a prime $\overline{S_{1,1,2}}$-free chordal graph then it is either a
$\overline{2P_1+P_2}$-free graph or a thick spider.
\end{lemma}

\begin{proof}
Let~$G$ be a prime $\overline{S_{1,1,2}}$-free chordal graph.
Note that since~$G$ is $\overline{S_{1,1,2}}$-free, it cannot contain $d$-$\mathbb{A}$ or
$d$-domino as an induced subgraph (see also \figurename~\ref{fig:diamond-prime-ext}).
If~$G$ is $\overline{P_1+P_4}$-free then, by
Lemma~\ref{lem:fig:diamond-prime-ext}, it must therefore be
$\overline{2P_1+P_2}$-free.

Now suppose that~$G$ contains an induced copy of $\overline{P_1+P_4}$.
Since~$G$ is prime,~$\overline{G}$ is also prime.
Furthermore,~$\overline{G}$ is $(2P_2,C_5,S_{1,1,2})$-free.
By Lemma~\ref{lem:cogemprimeext}, $\overline{G}$ must
contain one of the graphs in
\figurename~\ref{fig:co-gem-prime-ext}.
The only graph in \figurename~\ref{fig:co-gem-prime-ext} which is
$(2P_2,C_5,S_{1,1,2})$-free is the net, so~$\overline{G}$ must contain a
net.
By Lemma~\ref{lem:thinspider},~$\overline{G}$ is a thin spider, so~$G$ is a
thick spider, completing the proof.\qed
\end{proof}

As a corollary of the above lemma, we get the following:
\begin{theorem}\label{thm:co-chair}
Every $\overline{S_{1,1,2}}$-free chordal graph has clique-width at most~$4$.
\end{theorem}

\begin{proof}
Let~$G$ be an $\overline{S_{1,1,2}}$-free chordal graph.
By Lemma~\ref{lem:prime}, we may assume that~$G$ is prime.
If~$G$ is $\overline{2P_1+P_2}$-free then it has clique-width at most~$3$ by
Lemma~\ref{lem:gem-chordal}. By Lemma~\ref{lem:primecochairfrchordal}, we may therefore assume that~$G$ is a thick spider.
Note that since a thick spider is the complement of a thin spider (see also the 
definition of a \hyperref[def:thin-spider]{thin spider}),~$K$ is an independent set,~$I$ is a clique and~$R$ is complete to~$I$ and anti-complete to~$K$.
Every vertex in~$K$ has exactly one non-neighbour in~$I$ and vice versa. Since~$G$ is prime and~$R$ is a module,~$R$ contains at most one vertex.

Let $i_1,\ldots,i_p$ be the vertices in~$I$ and let $k_1,\ldots,k_p$ be the
vertices in~$K$ such that for each~$j \in \{1,\ldots,p\}$, the vertex~$i_j$ is the unique
non-neighbour of~$k_j$ in~$I$.
Let~$G_j$ be the labelled copy of $G[i_1,\ldots,i_j,k_1,\ldots,k_j]$ where every~$i_h$ is labelled~$1$ and every~$k_h$ is labelled~$2$.
Now $G_1 = 1(i_1) \oplus 2(k_1)$ and for $j \in \{1,\ldots,p-1\}$ we can construct~$G_{j+1}$ from~$G_j$ as follows:
$$
G_{j+1}=\rho_{3\rightarrow 1}(\rho_{4\rightarrow 2}(\eta_{1,3}(\eta_{1,4}(\eta_{2,3}(G_j\oplus 3(i_{j+1})\oplus 4(k_{j+1})))))).
$$
If~$R = \emptyset$ then using the above recursively we get a 4-expression for~$G_p$ and therefore for~$G$.
If~$R=\{x\}$ then we obtain a 4-expression for~$G$ using $\eta_{1,4}(G_p \oplus\nobreak 4(x))$.
Therefore~$G$ indeed has clique-width at most~$4$.
This completes the proof.\qed
\end{proof}

\section{The Classifications}\label{sec:chordal-classification}

In this section we first prove our main result, Theorem~\ref{thm:chordal-classification}, which was presented in Section~\ref{sec:intro}. 
Recall that~$F_1$ and~$F_2$ are the graphs shown in \figurename~\ref{fig:open-chordal}.

\medskip
\noindent
\faketheorem{Theorem~\ref{thm:chordal-classification} (restated).}
{\em Let~$H$ be a graph with $H\notin \{F_1,F_2\}$. The class of $H$-free chordal graphs has bounded clique-width if and only if
\begin{itemize}
\item $H=K_r$ for some $r\geq 1$;
\item $H\ssi \bull$;
\item $H\ssi P_1+P_4$;
\item $H\ssi \overline{P_1+P_4}$;
\item $H\ssi \overline{K_{1,3}+2P_1}$;
\item $H\ssi P_1+\overline{P_1+P_3}$;
\item $H\ssi P_1+\overline{2P_1+P_2}$ or
\item $H\ssi \overline{S_{1,1,2}}$.
\end{itemize}
}

\begin{proof}
Let~$H$ be a graph with $H \notin \{F_1,F_2\}$.
If $H=K_r$ for some $r\geq 1$ then we use Lemma~\ref{lem:clique-chordal}.
If~$H$ is an induced subgraph of a graph in $\{\bull,\allowbreak P_1+\nobreak
P_4,\allowbreak \overline{P_1+P_4}\}$ then we use
Lemmas~\ref{lem:bull-chordal},~\ref{lem:cogem-chordal}
or~\ref{lem:gem-chordal}, respectively.
If~$H$ is an induced subgraph of a graph in $\{\overline{K_{1,3}+2P_1},\allowbreak 
P_1+\nobreak \overline{P_1+P_3},\allowbreak P_1+\nobreak \overline{2P_1+P_2},\allowbreak \overline{S_{1,1,2}}\}$, then we use
Theorems~\ref{thm:co-k13+2p1-chordal},~\ref{thm:p1+paw-chordal},~\ref{thm:p1+diamond-chordal} 
or~\ref{thm:co-chair}, respectively.

We now prove the reverse direction of the theorem.
Let $H\notin \{F_1,F_2\}$ be a graph such that the class of $H$-free chordal graphs has bounded clique-width.
We first prove two useful claims, which show that we are done in some special cases. 

\clm{\label{clm:f1-f2} If~$H$ is a proper induced subgraph of~$F_1$ or~$F_2$ then~$H$
is an induced subgraph of a graph in $\{\bull, \overline{K_{1,3}+2P_1}, P_1+\nobreak \overline{P_1+P_3},\allowbreak P_1+\nobreak \overline{2P_1+P_2},\allowbreak
\overline{S_{1,1,2}}\}$} 

\noindent
We prove Claim~\ref{clm:f1-f2} as follows.
Note that~$F_1$ and~$F_2$ are six-vertex graphs. The five-vertex induced
subgraphs of~$F_1$ are $\bull, \overline{K_{1,3}+P_1}$, and~$P_1+\nobreak
\overline{P_1+P_3}$. The five-vertex induced subgraphs of~$F_2$ are
$\bull,\allowbreak \overline{K_{1,3}+P_1},\allowbreak P_1+\nobreak \overline{2P_1+P_2},\allowbreak \overline{2P_1+P_3},\allowbreak 
\overline{S_{1,1,2}}$.
Since $\overline{K_{1,3}+P_1}$ and $\overline{2P_1+P_3}$ are induced subgraphs of~$\overline{K_{1,3}+2P_1}$, this completes the proof of the Claim~\ref{clm:f1-f2}.

\clm{\label{clm:6-vertex}
If~$H$ is an induced subgraph of a graph in 
$\{\overline{\bull+\nobreak P_1},\allowbreak F_3,\allowbreak Q, \overline{Q}\}$ (see \figurename~\ref{fig:clm:6-vertex})
then~$H$ must be an induced subgraph of a graph in
$\{
\bull,\allowbreak
P_1+\nobreak P_4,\allowbreak
\overline{P_1+P_4},\allowbreak
\overline{K_{1,3}+ 2P_1},\allowbreak 
P_1+\nobreak \overline{P_1+P_3},\allowbreak
\overline{S_{1,1,2}}\}$.}
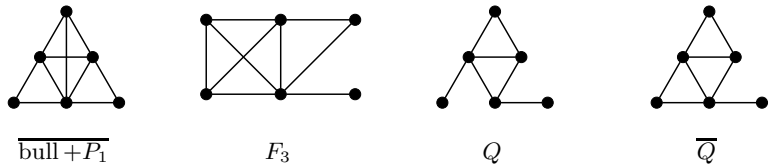
\begin{figure}
\begin{center}
\begin{tabular}{ccccc}
\begin{minipage}{0.21\textwidth}
\centering
\scalebox{0.7}{
{\begin{tikzpicture}[scale=1,rotate=30]
\GraphInit[vstyle=Simple]
\SetVertexSimple[MinSize=6pt]
\Vertex[a=0,d=0.57735026919]{a}
\Vertex[a=120,d=0.57735026919]{b}
\Vertex[a=240,d=0.57735026919]{c}
\Vertex[a=60,d=1.15470053838]{d}
\Vertex[a=180,d=1.15470053838]{e}
\Vertex[a=300,d=1.15470053838]{f}
\Edges(a,b,c,a,d,b,e)
\Edge(c)(f)
\Edge(c)(e)
\Edge(a)(f)
\Edge(c)(d)
\end{tikzpicture}}}
\end{minipage}
&
\begin{minipage}{0.21\textwidth}
\centering
\scalebox{0.7}{
{\begin{tikzpicture}[xscale=-1,rotate=135]
\GraphInit[vstyle=Simple]
\SetVertexSimple[MinSize=6pt]
\Vertices{circle}{a,b,c,d}
\Vertex[x=1,y=2]{y}
\Vertex[x=2,y=1]{z}
\Edges(a,b,c,d,a,c)
\Edges(b,d)
\Edges(a,z,b,y)
\end{tikzpicture}}}
\end{minipage}
&
\begin{minipage}{0.21\textwidth}
\centering
\scalebox{0.7}{
{\begin{tikzpicture}[scale=1,rotate=30]
\GraphInit[vstyle=Simple]
\SetVertexSimple[MinSize=6pt]
\Vertex[a=0,d=0.57735026919]{a}
\Vertex[a=120,d=0.57735026919]{b}
\Vertex[a=240,d=0.57735026919]{c}
\Vertex[a=60,d=1.15470053838]{d}
\Vertex[a=180,d=1.15470053838]{e}
\Vertex[a=300,d=1.15470053838]{f}
\Edges(a,b,c,a,d,b,e)
\Edge(c)(f)
\end{tikzpicture}}}
\end{minipage}
&
\begin{minipage}{0.21\textwidth}
\centering
\scalebox{0.7}{
{\begin{tikzpicture}[scale=1,rotate=30]
\GraphInit[vstyle=Simple]
\SetVertexSimple[MinSize=6pt]
\Vertex[a=0,d=0.57735026919]{a}
\Vertex[a=120,d=0.57735026919]{b}
\Vertex[a=240,d=0.57735026919]{c}
\Vertex[a=60,d=1.15470053838]{d}
\Vertex[a=180,d=1.15470053838]{e}
\Vertex[a=300,d=1.15470053838]{f}
\Edges(a,b,c,a,d,b,e)
\Edge(c)(f)
\Edge(c)(e)
\end{tikzpicture}}}
\end{minipage}\\
\\
$\overline{\bull+\nobreak P_1}$ &
$F_3$ & $Q$ & $\overline{Q}$
\end{tabular}
\caption{\label{fig:clm:6-vertex} The graphs $\overline{\bull+\nobreak P_1},\allowbreak F_3, Q$ and~$\overline{Q}$ from Claim~\ref{clm:6-vertex}.}
\end{center}
\end{figure}

\noindent
We prove Claim~\ref{clm:6-vertex} as follows.
If $H \in \{\overline{\bull+\nobreak P_1},\allowbreak F_3, Q, \overline{Q}\}$ then~$H$ contains an induced~$K_{1,3}$.
By Lemma~\ref{lem:claw-chordal}, since the class of~$H$-free chordal graphs has bounded clique-width, $H$ must be $K_{1,3}$-free.
Hence~$H$ must be a $K_{1,3}$-free induced subgraph of $\overline{\bull+\nobreak P_1},\allowbreak F_3, Q$ or~$\overline{Q}$.
We list the maximal $K_{1,3}$-free induced subgraphs of $\overline{\bull+\nobreak P_1},\allowbreak F_3, Q$ and~$\overline{Q}$, respectively, in
Table~\ref{t-tablefree}.
Since $\overline{K_{1,3}+P_1}$ and $\overline{2P_1+P_3}$ are induced subgraphs of~$\overline{K_{1,3}+2P_1}$, this completes the proof of Claim~\ref{clm:6-vertex}.

\begin{table}
\begin{center}
{
\setlength{\tabcolsep}{1em}
\begin{tabular}{c|c}
$H$ & Maximal $K_{1,3}$-free induced subgraphs of~$H$\\
\hline\\[-1em]
$\overline{\bull+\nobreak P_1}$ &
$\bull,\overline{P_1+P_4}, \overline{2P_1+ P_3}$\\
$F_3$ &
$\overline{K_{1,3}+ P_1}, P_1+\nobreak \overline{P_1+P_3}, \overline{2P_1+\nobreak P_3}$\\
$Q$ &
$\bull, P_1+\nobreak P_4, P_1+\nobreak \overline{P_1+P_3}, \overline{S_{1,1,2}}$\\
$\overline{Q}$ &
$\bull, \overline{P_1+P_4},P_1+\nobreak \overline{P_1+P_3},\overline{S_{1,1,2}}$
\end{tabular}
}
\end{center}
\caption{The maximal $K_{1,3}$-free induced subgraphs of $\overline{\bull+P_1},F_3,Q$ and~$\overline{Q}$. }\label{t-tablefree}
\end{table}

\medskip
\noindent
Due to Claims~\ref{clm:f1-f2} and~\ref{clm:6-vertex},
if~$H$ is an induced subgraph of a graph in
$\{\overline{\bull+\nobreak P_1},\allowbreak F_1,\allowbreak F_2,\allowbreak F_3,\allowbreak Q, \overline{Q}\}$ then we are done.

Since the class of split graphs is contained
in the class of chordal graphs, the class of $H$-free split graphs must also
have bounded clique-width. By Lemma~\ref{lem:split-classification}, the
graph~$H$ must therefore be a clique, an independent set or an induced subgraph of
a graph in $\{F_4,\overline{F_4},F_5,\overline{F_5}\}$ (see \figurename~\ref{fig:open-split}).
If~$H$ is a clique then we are done.
If~$H$ is an independent set then Lemma~\ref{lem:4p1-chordal} tells us that~$H$ can have at most three
vertices, in which case~$H$ is an induced subgraph of the $\bull$ and we are done.
We may therefore assume that~$H$ is an induced
subgraph of a graph in $\{F_4,\overline{F_4},F_5,\overline{F_5}\}$ and we will consider each of these possibilities in turn.
Furthermore, $H$ must be 
$4P_1$-free and $K_{1,3}$-free, 
otherwise the clique-width of
$H$-free chordal graphs would be unbounded 
(by Lemmas~\ref{lem:4p1-chordal} and~\ref{lem:claw-chordal}, respectively).

\thmcase{$H \ssi F_4$.}
Since~$F_4$ contains an independent set on five vertices and~$H$ is $4P_1$-free, two of these vertices
must be deleted in~$H$. Therefore~$H$ must be an induced subgraph of $\bull,\allowbreak
P_1+\nobreak P_4,\allowbreak P_1+\nobreak \overline{P_1+P_3},\allowbreak P_1+\nobreak \overline{2P_1+P_2}$ or
$\overline{P_1+\overline{P_1+P_3}}$. In the first four cases we are done immediately. 
The graph $\overline{P_1+\overline{P_1+P_3}}$ (also known as the dart) is
an induced subgraph of~$F_3$, so in the fifth case we are done by
Claim~\ref{clm:6-vertex}.

\thmcase{$H \ssi \overline{F_4}$.}
The graph~$\overline{F_4}$ contains two induced copies of~$K_{1,3}$ (which are not vertex-disjoint). 
Since~$H$ is $K_{1,3}$-free, it follows that~$H$ is an induced subgraph of
$F_1, \overline{K_{1,3}+2P_1}$ or $\overline{P_1+P_4}$.
In the first case, we are done by Claim~\ref{clm:f1-f2}.
In the other two cases we are done immediately.

\thmcase{$H \ssi F_5$.}
Since~$F_5$ contains an independent set on four vertices, one of these vertices
must be deleted in~$H$. Therefore~$H$ must be an induced subgraph of $F_1, F_2, F_3$ or~$Q$.
In the first two cases we apply Claim~\ref{clm:f1-f2} and in the other two we apply Claim~\ref{clm:6-vertex}.

\thmcase{$H \ssi \overline{F_5}$.}
Since~$\overline{F_5}$ contains an independent set on four vertices, one of these vertices
must be deleted in~$H$. Therefore~$H$ must be an induced subgraph of $\overline{\bull+\nobreak P_1}, F_2, F_3$ or~$\overline{Q}$.
In each of these cases, we are done by Claims~\ref{clm:f1-f2} or~\ref{clm:6-vertex}.

This completes the proof of Theorem~\ref{thm:chordal-classification}.\qed
\end{proof}

\noindent
We now prove our dichotomy for $H$-free weakly chordal graphs, which we recall below.

\medskip
\noindent
\faketheorem{Theorem~\ref{t-weakly-chordal} (restated).}
{\em Let~$H$ be a graph. The class of
$H$-free weakly chordal graphs has bounded clique-width if and only if~$H$ is an induced subgraph of~$P_4$.
}

\begin{proof}
Let~$H$ be a graph. First suppose that~$H$ is an induced subgraph of~$P_4$. 
Then the class of $H$-free weakly chordal graphs is contained in the class of $P_4$-free graphs,
which have bounded clique-width by Lemma~\ref{l-p4}.
Now suppose that~$H$ is not an induced subgraph of~$P_4$. 
Below we show that the class of $H$-free weakly chordal graphs has unbounded clique-width.

Suppose that~$H$ is not a split graph. Then the class of $H$-free weakly chordal graphs contains the class of split graphs, which has unbounded clique-width by Lemma~\ref{lem:split-chordal}
(or Lemma~\ref{lem:split-classification}).
From now on assume that~$H$ is a split graph.
Suppose that~$H$ contains a cycle~$C$. As~$H$ is a split graph, it is $(C_4,C_5,2P_2)$-free by Lemma~\ref{lem:split}.
Hence,~$C$ is isomorphic to~$C_3$. Then the class of $H$-free weakly chordal graphs contains the class
of bipartite weakly chordal graphs, which contains the class of bipartite permutation graphs, which has unbounded clique-width by Lemma~\ref{l-bippermut}.
From now on assume that~$H$ contains no cycle.

We claim that~$H$ has an induced~$3P_1$. For contradiction, suppose~$H$ is $3P_1$-free.
Then every connected component of~$H$ is a path. As~$H$ is $3P_1$-free,~$H$ has at most two connected components, each of which is a path on at most four vertices. Because~$H$ is not an induced subgraph of~$P_4$, this means that~$H$ has exactly two connected components. As~$H$ is $3P_1$-free, each of
these components is a path on at most two vertices.
As~$H$ is $2P_2$-free, at most one of the components contains an edge. However, then~$H$ is an induced subgraph of~$P_4$, a contradiction.
Now, as~$H$ has an induced~$3P_1$, the class of complements of $H$-free weakly chordal graphs contains the class of $C_3$-free weakly chordal graphs, which has unbounded clique-width, as shown above.
Applying Fact~\ref{fact:comp} completes the proof.\qed
\end{proof}

\section{An Application}\label{sec:consec}

In this section we give an application of Theorem~\ref{thm:chordal-classification} by showing how to use it to prove that the class of $(K_4,\allowbreak 2P_1+\nobreak P_3$)-free graphs has bounded clique-width (see also \figurename~\ref{fig:2-forb-case}), which means that only~$13$ (non-equivalent) cases 
remain open~\cite{DP15}.
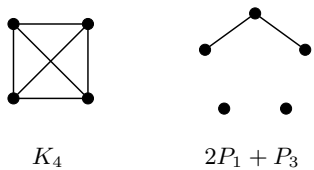
\begin{figure}
\begin{center}
\begin{tabular}{cc}
\begin{minipage}{0.20\textwidth}
\centering
\scalebox{0.7}{
{\begin{tikzpicture}[scale=1,rotate=45]
\GraphInit[vstyle=Simple]
\SetVertexSimple[MinSize=6pt]
\Vertices{circle}{a,b,c,d}
\Edges(a,b,c,d,a,c)
\Edges(b,d)
\end{tikzpicture}}}
\end{minipage}
&
\begin{minipage}{0.20\textwidth}
\centering
\scalebox{0.7}{
{\begin{tikzpicture}[scale=1,rotate=90]
\GraphInit[vstyle=Simple]
\SetVertexSimple[MinSize=6pt]
\Vertices{circle}{a,b,c,d,e}
\Edges(e,a,b)
\end{tikzpicture}}}
\end{minipage}\\
& \\
$K_4$ & $2P_1+P_3$
\end{tabular}
\end{center}
\caption{The graphs~$K_4$ and $2P_1+\nobreak P_3$.}\label{fig:2-forb-case}
\end{figure}

\begin{theorem}\label{thm:K4-2p1+p3-chordal}
The class of $(K_4,\allowbreak 2P_1+\nobreak P_3)$-free graphs has bounded clique-width.
\end{theorem}

\begin{proof}
Suppose~$G$ is a $(K_4,\allowbreak 2P_1+\nobreak P_3)$-free graph. If~$G$ is chordal then it is a
$K_4$-free chordal graph, in which case it has bounded clique-width by
Lemma~\ref{lem:clique-chordal}. We may therefore assume that~$G$ contains
an induced cycle~$C$ with
vertices $v_1,v_2,\ldots,v_k$ in that order, such that $k \geq 4$. We may also
assume that this induced cycle is chosen such that~$k$ is minimal. Note that $k \leq
7$, otherwise $G[v_1,v_3,v_5,v_6,v_7]$ would be a $2P_1+\nobreak P_3$.

We partition the vertices not on the cycle~$C$ as follows. For $S \subseteq
\{1,\ldots,k\}$, let~$V_S$ contain those vertices $x \in V(G) \setminus C$ such
that $N_C(x) = \{v_i \; | \; i \in S\}$. We say that a set~$V_S$ is {\em large}
if it contains at least seven vertices, otherwise we say that it is {\em
small}. We now prove some useful properties about these sets.

\clm{\label{clm:K4-2p1+p3:no-4p1}
Suppose $S, T \subseteq \{1,\ldots,k\}$ with $S \neq T$. If $x,x' \in V_S$
and $y,y' \in V_T$ then $G[x,x',y,y']$ is not a~$4P_1$.}
Indeed, suppose that $G[x,x',y,y']$ is a~$4P_1$. Without loss of generality, we
may assume $i \in T \setminus S$. Then $G[x,x',y,v_i,y']$ is a $2P_1+\nobreak P_3$.

\clm{\label{clm:K4-2p1+p3:V12-indep}
If~$v_i$ and~$v_j$ are consecutive vertices of the cycle and $\{i,j\}
\subseteq S \subseteq \{1,\ldots,k\}$ then~$G[V_S]$ is independent.}
Indeed, if $x,x' \in S$ were adjacent then $G[x,x',v_i,v_j]$ would be a~$K_4$.

\clm{\label{clm:K4-2p1+p3:s-bigger-than-1}
Suppose $S\subseteq \{1,\ldots,k\}$. If $|S| \leq 1$ then~$V_S$ is small.}
Indeed, suppose $S= \emptyset$ or $S=\{1\}$. If $x,y \in V_S$ then $x,y$ must
be adjacent, otherwise $G[x,y,v_2,v_3,v_4]$ would be a $2P_1+\nobreak P_3$. Therefore~$V_S$ is a clique in~$G$. Since~$G$ is $K_4$-free, $|V_S| \leq 3$.

\clm{\label{clm:K4-2p1+p3:2-indep-V_S}
Suppose $S,T \subseteq \{1,\ldots,k\}$ with $S\neq T$.
If~$V_S$ and~$V_T$ are independent sets in~$G$ and~$V_T$ is large then at most one vertex of~$V_S$ has
more than one non-neighbour in~$V_T$.}
Indeed, since $|V_T| \geq 7 \geq 4$, by Claim~\ref{clm:K4-2p1+p3:no-4p1} for
any pair of vertices $x,x' \in V_S$, at least one of these vertices must have
at least two neighbours in~$V_T$.
Therefore every vertex of~$V_S$ except perhaps one has at least two neighbours in~$V_T$.
Consider a vertex $x \in V_S$ that has two neighbours $y,y' \in V_T$. The vertex~$x$ cannot have two non-neighbours $z,z' \in V_T$,
otherwise $G[z,z',y,x,y']$ would be a $2P_1+\nobreak P_3$.
Therefore every vertex of~$V_S$ except perhaps one has at most one non-neighbour in~$V_T$. Hence,
at most one vertex of~$V_S$ has more than one non-neighbour in~$V_T$.

\clm{\label{clm:K4-2p1+p3:3-indep-V_S}
Suppose $S,T,U \subseteq \{1,\ldots,k\}$ are pairwise distinct.
If $V_S, V_T$ and~$V_U$ are independent sets in~$G$ then $G[V_S \cup V_T \cup V_U]$ has bounded
clique-width.}
Indeed, if any set in $\{V_S, V_T, V_U\}$ is small then by
Fact~\ref{fact:del-vert} we may assume it is empty. By
Claim~\ref{clm:K4-2p1+p3:2-indep-V_S} and Fact~\ref{fact:del-vert}, we may
delete at most two vertices from each of $V_S, V_T, V_U$ after which every
vertex in each of these sets will have at most one non-neighbour in each of the
other two sets. In other words, every vertex in one of these sets will have at most
two non-neighbours in total in the other two sets. Applying a bipartite
complementation between each pair of sets (which we may do by
Fact~\ref{fact:bip}) yields a graph of maximum degree at most~2. This graph has
bounded clique-width by Lemma~\ref{lem:atmost-2}.

\clm{\label{clm:K4-2p1+p3:4-indep-V_S}
Suppose $R,S,T,U \subseteq \{1,\ldots,k\}$ are pairwise distinct.
If $V_R,V_S,V_T,V_U$ are all independent sets in~$G$ then at least one of $V_R,V_S, V_T, V_U$
is small.}
Indeed, suppose for contradiction that all of $V_R,V_S, V_T, V_U$ are large.
Let $V'_R,\allowbreak V'_S,\allowbreak V'_T$ and~$V'_U$ be the sets of those vertices in $V_R,V_S, V_T$ and~$V_U$, respectively, that do not have two non-neighbours in any of the three
other sets. By Claim~\ref{clm:K4-2p1+p3:2-indep-V_S}, each of
$V'_R,V'_S, V'_T$ and~$V'_U$
has at least $7-3=4$ vertices. Let $r \in V'_R$. Since $|V'_S| \geq 2$,
there must be a vertex $s \in V'_S$ adjacent to~$r$. Since $|V'_T|
\geq 3$, there must be a vertex $t \in V'_T$ adjacent to~$r$ and~$s$.
Since $|V'_U| \geq 4$, there must be a vertex $u \in V'_U$ adjacent to
$r,s$ and~$t$. Now $G[r,s,t,u]$ is a~$K_4$, a contradiction.

\medskip
If any set~$V_S$ is small then, by Fact~\ref{fact:del-vert}, we may assume it is
empty. We may therefore assume that every set~$V_S$ is either large or empty.
Furthermore, we may assume that some large
set~$V_S$ is not an independent set,
otherwise we can apply Claim~\ref{clm:K4-2p1+p3:4-indep-V_S},
to find that at most three sets~$V_S$ are non-empty and then,
after deleting the $k\leq 7$ vertices of~$C$
(which we may do by Fact~\ref{fact:del-vert}),
we can apply Claim~\ref{clm:K4-2p1+p3:3-indep-V_S}
to find that the clique-width of~$G$ is bounded.

We claim that $k=4$. For contradiction, suppose that $5\leq k\leq 7$.
Let $S \subseteq \{1,\ldots,k\}$ be a set such that~$G[V_S]$ is large and not independent.
By Claim~\ref{clm:K4-2p1+p3:s-bigger-than-1}, it follows that $|S|
\geq 2$. By Claim~\ref{clm:K4-2p1+p3:V12-indep}, the vertices of~$V_S$ cannot
be adjacent to two consecutive vertices of~$C$.
Without loss of generality, assume that $1 \in S$, which implies that $2, k \notin S$. Then there must be
a number $j \in \{3,\ldots,k-1\}$ such that $j \in S$, and $2,\ldots,j-1 \notin
S$. If $j \leq k-2$ then choosing $x \in V_S$ we find that
$G[x,v_1,\ldots,v_j]$ is a~$C_{j+1}$, contradicting the minimality of~$k$. If
$j=k-1$ then choosing $x \in V_S$ we find that $G[v_{k-1},v_k,v_1,x]$ is a~$C_4$, contradicting the minimality of~$k$.
Hence, we conclude that indeed $k=4$.

Again, let $S \subseteq \{1,\ldots,k\}$ be a set such that~$G[V_S]$ is large and not independent.
By Claims~\ref{clm:K4-2p1+p3:V12-indep}
and~\ref{clm:K4-2p1+p3:s-bigger-than-1},
we find that $S=\{1,3\}$ or $S=\{2,4\}$.
If there exist vertices $x,y,z \in V_{\{1,3\}}$ that induce a~$P_3$ then $G[v_2,v_4,x,y,z]$ would be a
$2P_1+\nobreak P_3$, which is not possible. Therefore~$G[V_{\{1,3\}}]$ must be $P_3$-free, so it must be
a disjoint union of cliques. If~$G[V_{\{1,3\}}]$ contained a~$K_3$ on vertices
$x,y,z$ then $G[v_1,x,y,z]$ would be a~$K_4$, which is not possible. Thus every component of~$G[V_{\{1,3\}}]$ and (by symmetry)~$G[V_{\{2,4\}}]$ must be
isomorphic to either~$P_1$ or~$P_2$.

If~$G[V_{\{1,3\}}]$ and~$G[V_{\{2,4\}}]$ each contain at most one edge then,
by deleting at most one vertex from each of~$V_{\{1,3\}}$ and~$V_{\{2,4\}}$ (which
we may do by Fact~\ref{fact:del-vert}), we obtain a graph in which every
set~$V_S$ is independent, in which case we find that~$G$ has bounded clique-width by
proceeding as before: we first apply Claim~\ref{clm:K4-2p1+p3:4-indep-V_S}, then delete the vertices of~$C$ by Fact~\ref{fact:del-vert} and finally apply Claim~\ref{clm:K4-2p1+p3:3-indep-V_S}.
Without loss of generality, we may therefore assume that~$G[V_{\{1,3\}}]$ contains two
edges~$xx'$ and~$yy'$ (which together induce a~$2P_2$).

We claim that every set~$V_T$
other than~$V_{\{1,3\}}$ and~$V_{\{2,4\}}$ is empty. Indeed, for contradiction, suppose such a set~$V_T$ is non-empty. Then, as stated above,~$V_T$ must be independent and large.
By Claim~\ref{clm:K4-2p1+p3:s-bigger-than-1}, $|T| \geq 2$. By symmetry we may
therefore assume that $\{1,2\} \subseteq T$. If $z \in V_T$ is adjacent to
both~$x$ and~$x'$ then $G[x,x',v_1,z]$ would be a~$K_4$, which is not possible. Therefore any vertex
in~$V_T$ can be adjacent to at most one vertex in each of~$\{x,x'\}$ and
$\{y,y'\}$. Since $|V_T| \geq 7 \geq 5$,
we find that~$V_T$ contains two vertices~$z,z'$, which are not adjacent to each other (as~$V_T$ is independent) and which are both non-adjacent to
the same vertex in~$\{x,x'\}$ and to the same vertex in~$\{y,y'\}$. By Claim~\ref{clm:K4-2p1+p3:no-4p1}, this is a
contradiction, so~$V_T$ must indeed be empty.

Recall that by Fact~\ref{fact:del-vert} we may delete the four vertices of~$C$.
We are therefore reduced to proving that $G[V_{\{1,3\}} \cup V_{\{2,4\}}]$ has
bounded clique-width. Note that if $x \in V_{\{1,3\}}$ is non-adjacent to two
vertices~$y$ and~$y'$ in~$V_{\{2,4\}}$ then~$y$ and~$y'$ must be adjacent,
otherwise $G[y,y',v_1,x,v_3]$ would be a $2P_1+\nobreak P_3$ (which is not possible). This, together with the fact that~$G$ is $K_4$-free, implies that any vertex in~$V_{\{1,3\}}$ has at most
two non-neighbours in~$V_{\{2,4\}}$, and vice versa. Let~$G'$ be the graph obtained from
$G[V_{\{1,3\}} \cup V_{\{2,4\}}]$ by applying a bipartite complementation
between~$V_{\{1,3\}}$ and~$V_{\{2,4\}}$. Then~$G'$ has maximum degree
at most~3. By Fact~\ref{fact:bip}, it remains to show that
every connected component of~$G'$ has bounded clique-width.

Consider a connected component~$D$ of~$G'$.
We first prove that~$D$ contains at most four vertices of degree~3.
Let $x\in D$ be a vertex that has degree~3 in~$D$.
Without loss of generality assume that $x \in V_{\{1,3\}}$. Then~$x$ has two
neighbours~$y,y' \in V_{\{2,4\}}$ and one neighbour $x' \in V_{\{1,3\}}$.
Recall that~$y$ is adjacent to~$y'$ due to the fact that~$G$ is $2P_1+\nobreak P_3$-free.
For the same reason and because~$G[V_{\{1,3\}}]$ only has connected components isomorphic to~$P_1$ or~$P_2$, we find that~$y$ and~$y'$ are adjacent to~$x'$ in~$D$ if they have degree~3 in~$D$.
Hence either $V(D)=\{x,x',y,y'\}$ or~$y,y'$ each have degree~2 in~$D$ and~$x'$ is a cut-vertex of~$D$. In the first case,~$D$ has at most four vertices of degree~3.
In the second case, we note that~$x'$ is adjacent to neither~$y$ nor~$y'$ in~$D$ (otherwise, for the same reason as before,~$x'$ would be adjacent to both of them if it had degree~3 in~$D$, so~$V_D$ would only contain the vertices $x,x',y,y'$).
We then find that~$D$ is either obtained by identifying a vertex of a triangle and the end-vertex of a path, meaning that~$D$ has only one vertex of degree~3 (namely~$x$),
or else by connecting two vertex-disjoint triangles via
a path between one vertex of one triangle and one of the other, meaning that~$D$ has exactly two vertices of degree~3.

Because~$D$ has at most four vertices of degree~3, we may remove these vertices by Fact~\ref{fact:del-vert} and then apply Lemma~\ref{lem:atmost-2} to find that~$D$ has bounded clique-width. This completes the proof of Theorem~\ref{thm:K4-2p1+p3-chordal}.\qed
\end{proof}

\section{Concluding Remarks}\label{sec:conclusions}

In our main result we characterized all but two graphs~$H$ for which the class of $H$-free chordal graphs has bounded clique-width. In particular we identified four new graph classes of bounded clique-width, namely the classes of $H$-free chordal graphs with $H\in \{\overline{K_{1,3}+2P_1},\allowbreak P_1+\nobreak \overline{P_1+P_3},\allowbreak P_1+\nobreak \overline{2P_1+P_2},\allowbreak \overline{S_{1,1,2}}\}$.
We also showed that the restriction from $H$-free graphs to $H$-free perfect graphs does not yield any new classes of bounded clique-width.
Moreover, we determined a new class of $(H_1,H_2)$-free graphs, namely the class of $(K_4,2P_1+P_3)$-free graphs, that has bounded clique-width via a reduction to chordal graphs. The latter means that only the following~$13$ cases, up to an equivalence relation,\footnote{For graphs $H_1,\ldots, H_4$, the classes of
$(H_1,H_2)$-free graphs and $(H_3,H_4)$-free graphs are
equivalent if $\{H_3,H_4\}$ can be obtained from
$\{H_1,H_2\}$ by some combination of the two operations: complementing both graphs in the pair; or if one of the graphs in the pair is~$K_3$,
replacing it with $\overline{P_1+P_3}$ or vice versa.
If two classes are equivalent then one has bounded clique-width if and only if the other one does (see e.g.~\cite{DP15}).}
are open in the classification for $(H_1,H_2)$-free graphs (see~\cite{DP15}).
\begin{enumerate}
\item \label{oprob:twographs:3P_1} $H_1=3P_1, \overline{H_2} \in \{P_1+P_2+P_3,P_1+2P_2,P_1+P_5,P_1+S_{1,1,3},P_2+P_4,\allowbreak S_{1,2,2},\allowbreak S_{1,2,3}\}$;
\item \label{oprob:twographs:2P_1+P_2} $H_1=2P_1+P_2, \overline{H_2} \in \{P_1+P_2+P_3,P_1+2P_2,P_1+P_5\}$;
\item \label{oprob:twographs:P_1+P_4} $H_1=P_1+P_4, \overline{H_2} \in \{P_1+2P_2,P_2+P_3\}$ or
\item \label{oprob:twographs:2P_1+P_3} $H_1=\overline{H_2}=2P_1+P_3$.
\end{enumerate}

We identify the following three main directions for future work.

\medskip
\noindent
{\em 1. Determine whether or not the class of $H$-free chordal graphs has bounded clique-width when $H \in \{F_1,F_2\}$.}

\medskip
\noindent
For this purpose, we recently managed to show that the class of
$H$-free split graphs has bounded clique-width in both these cases~\cite{BDHP15b} and we are currently exploring whether it is
possible to generalize the proof of this result to the class of $H$-free
chordal graphs. 
This seems to be a challenging task, as clique-width has a subtle transition from bounded to unbounded
even if the class of graphs under consideration has a ``slight'' enlargement. 
For instance, we showed that the class of $(P_1+\nobreak \overline{P_1+P_3})$-free chordal graphs has bounded clique-width, whereas 
the class of $(P_1+\nobreak \overline{2P_1+P_3})$-free chordal graphs, or even
$(2P_1+\nobreak \overline{3P_1})$-free split graphs (see Lemma~\ref{lem:split-classification}) already has unbounded clique-width. 

\medskip
\noindent
{\em 2. Exploit 
the techniques developed in this paper to attack
some of the other open cases in the classification for $(H_1,H_2)$-free graphs.}

\medskip
\noindent
In particular the case $H_1=2P_1+\nobreak P_3$, $H_2=\overline{2P_1+P_3}$ seems a good candidate for
a possible proof of bounded clique-width via a reduction to $\overline{2P_1+P_3}$-free chordal graphs (this subclass
of chordal graphs has bounded clique-width by Theorem~\ref{thm:chordal-classification}).
For this direction we also note that it may be worthwhile to more closely
examine the relationship between our study and the one on the computational
complexity of the {\sc Graph Isomorphism} problem (GI) for classes of
$(H_1,H_2)$-free graphs, which was initiated by Kratsch and Schweitzer~\cite{KS12}.
Recently, Schweitzer~\cite{Sc15} proved that for this study the number of open
cases is finite and pointed out similarities between classifying
boundedness of clique-width and solving GI for special graph classes.
Indeed, Grohe and Schweitzer~\cite{GS15} recently proved that {\sc Graph Isomorphism} is polynomial-time solvable on graphs of bounded clique-width.

\medskip
\noindent
{\em 3. Determine
whether or not the class of $H$-free split graphs has bounded clique-width when $H \in \{F_4,F_5\}$.}

\medskip
\noindent
The
fact that the (un)boundedness of the clique-width of the class of
$H$-free split graphs is known for so many graphs~$H$ raises the question
whether we can obtain a full classification of all graphs~$H$ for which the
class of $H$-free split graphs has bounded clique-width.
We recently 
reduced~\cite{BDHP15b} this to two problematic cases, namely the graphs~$F_4$ and~$F_5$ displayed in Figure~\ref{fig:open-split}.

\medskip
\noindent
Finally we pose the question of whether it is possible to extend the four newly found classes of $H$-free chordal graphs (when
$H\in \{\overline{K_{1,3}+2P_1},\allowbreak P_1+\nobreak \overline{P_1+P_3},\allowbreak P_1+\nobreak \overline{2P_1+P_2},\allowbreak \overline{S_{1,1,2}}\}$) to larger classes
of graphs for which {\sc Hamilton Cycle} is polynomial-time solvable.
\bibliography{mybib-chordal}

\end{document}